\documentclass{article}
\usepackage{graphicx} %
\usepackage{subcaption}

\usepackage{fullpage} %

\usepackage{amsmath}
\usepackage{amsfonts}
\usepackage{amssymb}
\usepackage{amsthm}
\usepackage{hyperref}
\usepackage[utf8]{inputenc}

\usepackage{bm}

\usepackage{multirow}

\usepackage{cleveref}

\usepackage[section]{algorithm}

\newtheorem{theorem}{Theorem}[section]
\newtheorem{lemma}[theorem]{Lemma}

\newtheorem{proposition}[theorem]{Proposition}

\newtheorem{definition}[theorem]{Definition}

\newtheorem{problem}[theorem]{Problem}%

\usepackage{verbatim, comment, color}
\usepackage{enumerate}
\usepackage{setspace}

\usepackage[dvipsnames]{xcolor}

\newcommand{\be}{\begin{equation}}
\newcommand{\ee}{\end{equation}}

\newcommand{\cP}{{\cal P}}
\newcommand{\cX}{{\cal X}}

\newcommand{\R}{{\mathbb R}}
\newcommand{\N}{{\mathbb N}}
\newcommand{\E}{\mathbb{E}}

\DeclareMathOperator{\proj}{proj}

\newcommand{\q}{\quad}

\newcommand{\nn}{\nonumber}

\numberwithin{equation}{section}
\numberwithin{theorem}{section}

\usepackage{authblk}

\title{%
Dual Attainment in Multi-Period Multi-Asset Martingale Optimal Transport and Its Computation}

\author[1]{Charlie Che\thanks{Authors are listed in alphabetical order by last name.}\thanks{charlie.che@jpmchase.com}}

\author[2]{Tongseok Lim$^\ast$\thanks{lim336@purdue.edu}}
\author[3]{Yue Sun$^\ast$\thanks{yue.sun@jpmorgan.com}}

\affil[1]{Quantitative Trading \& Research, JPMorganChase, New York, NY 10017, USA}
\affil[2]{Mitch Daniels School of Business, Purdue University, West Lafayette, Indiana 47907, USA}
\affil[3]{Global Technology Applied Research, JPMorganChase, New York, NY 10001, USA}

\begin{document}

\maketitle

\begin{abstract} 
We establish dual attainment for the multimarginal, multi-asset martingale optimal transport (MOT) problem, a fundamental question in the mathematical theory of model-independent pricing and hedging in quantitative finance. Our main result proves the existence of dual optimizers under mild regularity and irreducibility conditions, extending previous duality and attainment results from the classical and two-marginal settings to arbitrary numbers of assets and time periods. This theoretical advance provides a rigorous foundation for robust pricing and hedging of complex, path-dependent financial derivatives. To support our analysis, we present numerical experiments that demonstrate the practical solvability of large-scale discrete MOT problems using the state-of-the-art primal-dual linear programming (PDLP) algorithm. In particular, we solve multi-dimensional (or vectorial) MOT instances arising from the robust pricing of worst-of autocallable options, confirming the accuracy and feasibility of our theoretical results. Our work advances the mathematical understanding of MOT and highlights its relevance for robust financial engineering in high-dimensional and model-uncertain environments. 
\end{abstract}

\section{Introduction}
\label{sec:intro}
Martingale optimal transport (MOT) has emerged as a central topic in mathematical finance and probability, providing a rigorous framework for model-independent pricing and hedging of financial derivatives under marginal and martingale constraints. The MOT problem generalizes classical optimal transport (OT) by requiring the transport plan to respect the martingale property, which is fundamental for risk-neutral valuation in quantitative finance. While the classical OT problem, introduced by Monge and Kantorovich~\cite{monge1781memoire,kantorovich1942translocation}, has a rich mathematical theory and wide-ranging applications~\cite{Vi09,Sa15}, the martingale extension introduces significant new challenges, both theoretically and computationally.

A key question in MOT is the existence of dual optimizers, the so-called \emph{dual attainment} property. Dual attainment is not only of intrinsic mathematical interest, but also has direct implications for robust hedging: dual optimizers correspond to explicit sub- and super-hedging strategies using tradable instruments~\cite{BeHePe11,bj,bnt,blo}. While duality and dual attainment have been extensively studied in the classical OT and single-asset, two-period MOT settings~\cite{bj,bnt,blo,eglo}, the general case of multi-marginal, multi-asset MOT remains much less understood. Theoretical advances in this direction are crucial for the robust pricing and hedging of complex, path-dependent financial products, which often depend on the joint evolution of multiple assets over several time periods.

This paper is devoted to the theoretical development of dual attainment for the multimarginal, multi-asset MOT problem. We establish the existence of dual optimizers under mild regularity and irreducibility conditions, extending previous results to arbitrary numbers of assets and time periods. Our approach builds on and generalizes the techniques developed for the single-asset, two-period case, and provides a rigorous foundation for robust pricing and hedging in high-dimensional, path-dependent settings. The main results contribute to the mathematical understanding of MOT and open the door to further applications in quantitative finance.

While the primary focus of this work is theoretical, we complement our analysis with a numerical demonstration. Specifically, we apply a state-of-the-art large-scale linear programming solver, PDLP~\cite{applegate2021practical,hinder2023worst}, to solve discrete instances of the multimarginal MOT problem. To our knowledge, this is the first demonstration of PDLP applied to large-scale multimarginal MOT, and the results provide supporting evidence for the practical relevance and accuracy of our theoretical findings. The numerical experiments, including a case study on a worst-of autocallable option, are intended to illustrate the feasibility of solving high-dimensional MOT problems and to validate the dual attainment results in realistic financial scenarios. %

\paragraph{Contributions.}  The primary contribution of this paper is a rigorous proof of dual attainment for the multimarginal, multi-asset martingale optimal transport (MOT) problem under general conditions. We also present numerical illustrations based on PDLP as supporting evidence; however, the development of computational methodologies is not a focus of this work and serves only to corroborate the theoretical results. Overall, our findings advance the mathematical theory of MOT and provide a foundation for robust pricing and hedging in complex financial markets.

\paragraph{Organization.}
The remainder of the paper is organized as follows. 
Section~\ref{sec:lit-review} provides a review of the related literature on optimal transport, martingale optimal transport, and their connections to mathematical finance. 
Section~\ref{sec:financial-applications} discusses key applications of MOT in quantitative finance, motivating the theoretical developments of this work.
Section~\ref{MTintro} introduces the mathematical formulation of the multimarginal martingale optimal transport problem, including the market setting, the VMOT problem, and its dual.
Section~\ref{sec:theoretical-results} develops our main theoretical results, establishing duality and dual attainment in the multi-marginal, multi-asset setting, and highlights the key ideas underlying the proofs. 
Section~\ref{sec:numerical-methods} describes the numerical methodology, including the PDLP algorithm and the discrete MOT formulation, and Section~\ref{sec:numerical-results} reports numerical experiments that support the theoretical findings, with a focus on a worst-of autocallable option case study.
Section~\ref{sec:conclusion} concludes the paper and outlines directions for future research.
Finally, Section~\ref{sec:proofs} contains detailed proofs of the main theoretical results.

\subsection{Related Literature}
\label{sec:lit-review}

The classical OT problem has been extensively studied, with foundational contributions by Monge, Kantorovich, and later Villani~\cite{Vi09} and Santambrogio~\cite{Sa15}. The extension to martingale constraints was pioneered in~\cite{BeHePe11,bj}, with further developments on duality and dual attainment in~\cite{bnt,blo,eglo}. Multi-marginal and multi-asset MOT has been investigated in~\cite{cot19,nst20,GKL2,Lim23}, but dual attainment in the general setting remains an open and challenging problem. Applications in quantitative finance include model calibration~\cite{dl19,G20,glw19}, robust pricing and hedging~\cite{Ho11,eglo}, and the extraction of risk-neutral marginals from market data using the Breeden--Litzenberger formula~\cite{bl78}. Moreover, recent advances in large-scale linear programming, such as PDLP~\cite{applegate2021practical,hinder2023worst}, have enabled practical solution of high-dimensional MOT problems, which we adopt as our computational approach in this paper.

\subsection{Applications in Quantitative Finance}
\label{sec:financial-applications}

The martingale optimal transport (MOT) framework has found increasing relevance in quantitative finance, both as a tool for model calibration and for robust pricing and hedging of complex derivatives. In this section, we review key application domains, highlight recent advances, and illustrate how the multi-marginal, multi-asset MOT approach enables new capabilities for practitioners.

\subsubsection{Model Calibration via Martingale Optimal Transport}

A major application of MOT in finance is the calibration of market models to observed data, particularly volatility surfaces and marginal distributions inferred from vanilla option prices. Traditional calibration methods often rely on parametric models, which may fail to capture the full range of market-implied dynamics. MOT provides a non-parametric alternative, allowing practitioners to fit models that are consistent with observed marginals while imposing the martingale property required by risk-neutral pricing.

Recent works have leveraged MOT for volatility surface fitting and stochastic volatility model calibration~\cite{dl19,G20,glw19}. For example, Demarch and Labordere formulated volatility surface fitting as a MOT problem using entropic regularization, while Guyon addressed the joint calibration of SPX and VIX smiles. Guo, Loeper, and Wang extended these ideas to a general setting, enabling calibration with arbitrary cost functionals and constraints. The flexibility of the MOT framework allows for the incorporation of additional market observable information, such as intermediate marginals or observable non vanilla instruments, to further tighten calibration and pricing bounds~\cite{lut-tightening-2019,sester-intermediate-2023}.

\subsubsection{Robust Pricing and Hedging}

Beyond calibration, MOT has emerged as a powerful tool for robust pricing and hedging of exotic derivatives. The pioneering work of Hobson~\cite{Ho11} used the Skorokhod embedding problem to derive tight, model-independent price bounds for complex derivatives. In practice, however, price bounds generated by classical optimal transport are often too wide to be useful, motivating the imposition of additional constraints such as martingality and observable market data.

The dual formulation of the MOT problem is particularly significant for hedging, as dual attainment corresponds to the existence of optimal sub- or super-hedging strategies using tradable instruments. Establishing dual attainment thus provides theoretical justification for robust hedging approaches and informs the construction of practical hedging portfolios.

Most research to date has focused on single-asset, two-period payoffs due to their theoretical tractability. However, many popular structured products in the market are multi-asset and path-dependent. Extending dual attainment results to the vectorial, multi-marginal setting, as in this work, provides a foundation for robust pricing and hedging of such exotics.

\subsubsection{Multi-Asset Path-Dependent Products}

The MOT framework is especially well-suited for pricing and hedging multi-asset, path-dependent derivatives, such as autocallables, worst-of options, and basket products. These instruments often depend on the joint evolution of several underlyings over multiple time periods, and their payoffs are sensitive to both marginal distributions and the dependence structure among assets.

A key advantage of the MOT approach is its minimal reliance on assumptions about asset correlations or joint dynamics. By only requiring consistency with observed marginals and the martingale property, MOT produces price bounds that reflect genuine model uncertainty. These bounds can be further tightened by incorporating additional market information, such as prices of liquid correlation-sensitive products or intermediate marginals. For example, the pricing bounds generated from the vectorial multi-marginal MOT framework can be interpreted as the range of model uncertainty due to unknown correlations. By introducing constraints based on market prices of pairwise correlation products (e.g., call vs. call or put vs. put), practitioners can further refine these bounds. This flexibility opens the door to robust pricing and risk management for a wide variety of exotic products.

In summary, the martingale optimal transport framework offers a robust foundation for a wide range of financial applications, from model calibration to pricing and hedging of complex derivatives. In the following section, we turn to the numerical implementation and results that demonstrate the practical effectiveness of these methods.

\section{Problem Formulation}
\label{MTintro}

This section introduces the mathematical framework for martingale optimal transport in multi-asset, multi-period markets. We first describe the market setting and the available marginal information, then formulate the vectorial martingale optimal transport (VMOT) problem, and finally present its dual formulation together with its financial interpretation.

\subsection{Market Setting and Marginal Information}

Consider a financial market with $d$ underlying assets, whose price processes are denoted by $(X_{t,i})_{t \ge 0}$ for each asset $i \in [d] := \{1,2,\dots,d\}$. We do not assume that the joint law of all asset prices—often referred to as the \emph{market model}—is known, since such information cannot be fully inferred from observable market data. 

Instead, following the classical argument of Breeden and Litzenberger~\cite{bl78}, we assume that the market reveals the marginal distribution of each asset price at any fixed maturity $t>0$. That is, for each $t$ and $i$, the distribution ${\rm Law}(X_{t,i}) = \mu_{t,i} \in \cP(\R)$ is observed, where $\cP(\cX)$ denotes the set of probability measures on a space $\cX$. We consider a finite collection of maturities $0 < T_1 < T_2 < \dots < T_N$, and for notational convenience write $X_{t,i} := X_{T_t,i}$ and $X_t := (X_{t,1},\dots,X_{t,d})$. Throughout the paper, we assume $(X_t)_{t \in [N]}$ is an $\R^d$-valued martingale, consistent with the standard risk-neutral framework in mathematical finance. While the full joint distribution ${\rm Law}(X_1,\dots,X_N)$ remains unspecified, we assume that the collection of $Nd$ marginal distributions $\mu_{t,i} := {\rm Law}(X_{t,i})$ is fixed and known. This motivates the following admissible class.

\begin{definition}[Vectorial Martingale Transports]
Assume the marginal measures $(\mu_{t,i})_{t,i}$ have finite first moments. Let $\mu_t := (\mu_{t,1},\dots,\mu_{t,d})$ and $\mu := (\mu_1,\dots,\mu_N)$. The set of \emph{vectorial martingale transports} is defined by
\begin{align}\label{VMT}
{\rm VMT}(\mu) := \Big\{ \pi \in \cP(\R^{Nd}) \ \Big| \ 
&\pi = {\rm Law}(X), \quad \E_\pi[X_{t+1} \mid X_t] = X_t, \\
&{\rm Law}(X_{t,i}) = \mu_{t,i} \ \text{for all } t \in [N],\ i \in [d] \Big\}. \nonumber
\end{align}
\end{definition}

\subsection{Vectorial Martingale Optimal Transport (VMOT)}

Given a measurable cost (or payoff) function $c : \R^{Nd} \to \R$, the vectorial martingale optimal transport (VMOT) problem is defined as
\begin{align}\label{VMOT}
\min_{\pi \in {\rm VMT}(\mu)} \ \E_\pi[c(X)] 
\qquad \text{or} \qquad
\max_{\pi \in {\rm VMT}(\mu)} \ \E_\pi[c(X)].
\end{align}

In a financial interpretation, $c(X)$ represents the payoff of a path-dependent derivative determined by the entire price trajectory $X=(X_t)_{t=1}^N$. Since the true market model $\pi$ is unknown, all martingale measures consistent with the observed marginal information must be considered. The minimum and maximum values in \eqref{VMOT} therefore correspond to the lower and upper arbitrage-free bounds for the derivative price.

A key distinction between VMOT and classical optimal transport lies in the martingale constraint $\E_\pi[X_{t+1} \mid X_t]=X_t$. This constraint imposes that, for each asset $i$ and time $t$, the marginal distributions satisfy the \emph{convex order} condition
\[
\mu_{t,i} \preceq_c \mu_{t+1,i},
\]
meaning that $\mu_{t,i}(f) \le \mu_{t+1,i}(f)$ for all convex functions $f$ on $\R$, where $\mu(f):= \int f(x) d\mu(x)$. Strassen’s theorem~\cite{St65} guarantees that this condition is both necessary and sufficient for ${\rm VMT}(\mu)$ to be nonempty. Accordingly, we assume $\mu_{t,i} \preceq_c \mu_{t+1,i}$ for all $t < N$ and $i \in [d]$ throughout the paper.

\subsection{Dual Problem and Financial Interpretation}

The VMOT problem is an infinite-dimensional linear program and therefore admits a natural dual formulation. When \eqref{VMOT} is a minimization problem, the dual problem takes the form
\begin{equation}\label{dualproblem}
\sup_{(\phi,h) \in \Psi} \ \mu(\phi),
\end{equation}
where $\phi = (\phi_1,\dots,\phi_N)$ with $\phi_t = (\phi_{t,1},\dots,\phi_{t,d})$, each $\phi_{t,i}:\R \to \R\cup\{-\infty\}$, and
\[
\mu(\phi) := \sum_{t=1}^N \sum_{i=1}^d \int \phi_{t,i}\,d\mu_{t,i}.
\]
The process $h=(h_1,\dots,h_N)$, with $h_N \equiv 0$, consists of functions $h_{t,i}:\R^{td}\to\R$ representing dynamic trading strategies. The admissible class $\Psi$ consists of all $(\phi,h)$ such that $\phi_{t,i} \in L^1(\mu_{t,i})$, $h_{t,i}$ is bounded, and the following pathwise inequality holds:
\begin{align}\label{ptwiseineq}
\sum_{t=1}^N \sum_{i=1}^d \phi_{t,i}(x_{t,i})
+ h_{t,i}(x_1,\dots,x_t)\,(x_{t+1,i}-x_{t,i})
\le c(x),
\end{align}
for all $x=(x_1,\dots,x_N)\in\R^{Nd}$.

For any $\pi\in{\rm VMT}(\mu)$ and $(\phi,h)\in\Psi$, the martingale property implies
\begin{equation}\label{hedgeprice}
\mu(\phi)
= \E_\pi\!\left[
\sum_{t=1}^N \sum_{i=1}^d \phi_{t,i}(X_{t,i})
+ h_{t,i}(X_1,\dots,X_t)\,(X_{t+1,i}-X_{t,i})
\right].
\end{equation}
Thus the dual value represents the cost of a semi-static hedging strategy composed of European options $\phi$ and dynamic trading strategies $h$. The inequality \eqref{ptwiseineq} enforces that the portfolio subhedges the payoff $c$ pathwise. When \eqref{VMOT} is a maximization problem, the dual problem becomes $\inf_{(\phi,h)\in\Psi} \mu(\phi)$, with the inequality reversed (and $\phi_{t,i}$ taking on their values in $\R \cup \{+\infty\}$).

Having established the market framework, the VMOT formulation, and its dual interpretation, we now turn to our main theoretical results concerning duality and dual attainment, which form the foundation for robust pricing and hedging in multi-asset, multi-period markets.

\section{Main Theoretical Results}  
\label{sec:theoretical-results}

In this section, we present the central theoretical contributions of the paper. We first recall the duality result for the vectorial martingale optimal transport (VMOT) problem, then discuss the challenges surrounding dual attainment and state our main theorem, and finally give an intuitive explanation of the proof strategy, emphasizing the role of irreducibility and local $L^1$ bounds.

\subsection{Duality in Martingale Optimal Transport}

Under mild assumptions on the cost function $c$ and the marginal distributions, the primal and dual optimal values coincide (see, e.g., \cite{eglo,Z}):
\begin{align}\label{duality}
P(c) 
&:= \inf_{\pi \in \mathrm{VMT}(\mu)} \mathbb{E}_\pi [c(X)] \\
&= \sup_{(\phi, h) \in \Psi} \mu(\phi) =: D(c). \nonumber
\end{align}
Furthermore, the primal problem is known to admit a minimizer; that is, there exists a martingale transport plan $\pi \in \mathrm{VMT}(\mu)$ such that $\mathbb{E}_\pi[c(X)] = P(c)$. Such martingale transports represent extremal market models for asset price evolution, optimizing the expected payoff of the derivative $c$.

\subsection{Dual Attainment: Existence and Challenges}

In contrast to the primal problem, whose solvability follows from standard compactness arguments under mild regularity conditions on $c$, proving \emph{dual attainment}—existence of an optimal dual pair $(\phi,h)$—is considerably more delicate. This difficulty already appears in classical optimal transport, as exemplified by Brenier’s work \cite{br}. In the martingale setting, the additional martingale constraint significantly complicates matters. As shown in \cite{bj,blo,bnt}, even in the single-asset, two-period case $(d,N)=(1,2)$, dual attainment may fail within the natural class $\Psi$, and positive results typically rely on careful analysis of limiting convex potentials.

This issue is not merely technical. In financial terms, dual optimizers correspond to optimal sub-/super-hedging strategies and encode important structural information about primal solutions, i.e., extremal market models. For this reason, the dual attainment problem has attracted substantial attention, though most of the literature focuses on the single-asset, two-period framework. Notable exceptions include \cite{cot19,nst20,os}, which treat multi-period single-asset settings, and \cite{d18-1,d18,dt19,GKL2,Lim23,os17}, which study vector-valued martingale transports mostly in two-period models.

\subsection{Main Result: Dual Attainment for Multi-Asset, Multi-Period MOT}

Our main theoretical result establishes dual attainment for the martingale optimal transport problem with an arbitrary number of time steps $N$ and assets $d$. This extension is particularly relevant for financial applications in which derivative payoffs depend on the full price path of multiple underlyings.

\begin{theorem} \label{main}
Let $(\mu_{t,i})_{t \in [N]}$ be an irreducible sequence of marginal distributions on $\mathbb{R}$ for each $i \in [d]$. Suppose $c:\R^{Nd}\to\R$ is a lower semicontinuous cost function such that
\[
|c(x)| \;\le\; \sum_{t=1}^N \sum_{i=1}^d v_{t,i}(x_{t,i})
\]
for some continuous functions $v_{t,i} \in L^1(\mu_{t,i})$. Then there exists a dual optimizer, that is, a family of functions
\[
(\phi,h) = (\phi_{t,i}, h_{t,i})_{t \in [N],\, i \in [d]}
\]
satisfying the pathwise inequality \eqref{ptwiseineq}, and such that, for every primal optimizer $\pi \in \mathrm{VMT}(\mu)$ solving \eqref{VMOT}, we have the pathwise equality
\begin{align}\label{ptwiseeq}
\sum_{t=1}^N \sum_{i=1}^d  \phi_{t,i}(x_{t,i})
+ h_{t,i}(x_1,\dots,x_t)\,(x_{t+1,i} - x_{t,i})
= c(x)
\quad \text{$\pi$-a.s.}
\end{align}
The dual optimizer need not belong to the class $\Psi$.
\end{theorem}

The pair $(\phi,h)$ in Theorem~\ref{main} is often called a \emph{dual maximizer}, as it solves the dual problem \eqref{dualproblem} associated with the primal minimization problem \eqref{VMOT}. The identity \eqref{ptwiseeq} shows that the semi-static portfolio built from $(\phi,h)$ \emph{replicates} the derivative $c$: for any primal minimizer $\pi$, the portfolio payoff agrees with $c(X)$ for $\pi$-almost every price path. The inequality \eqref{ptwiseineq} ensures that, away from primal minimizers, the same portfolio provides a pathwise subhedge of $c$. By symmetry, Theorem~\ref{main} also yields the existence of a dual \emph{minimizer} for the maximization problem in \eqref{VMOT} when $c$ is upper semicontinuous; in that case, the resulting portfolio superhedges the payoff and replicates $c$ along primal maximizers.

It is crucial to note that a dual optimizer is not guaranteed to lie in $\Psi$. Previous work has shown that the dual problem \eqref{dualproblem} is in general not attained within $\Psi$, even for $(d,N)=(1,2)$ (see \cite{bj,bnt}), unless $c$ satisfies additional regularity assumptions \cite{blo}. While $\Psi$ is a natural domain for formulating the dual problem, it is relatively ``narrow'' due to the lack of compactness in infinite dimensions. Proving dual attainment is therefore substantially harder than proving duality \eqref{duality}, which often follows from standard tools in functional analysis and convex duality. In particular, although the dual optimizer $(\phi,h)$ in Theorem~\ref{main} may lie outside $\Psi$, each component $\phi_{t,i}$ is real-valued $\mu_{t,i}$-a.s., and each $h_{t,i}$ is real-valued as well. Since the marginals $(\mu_{t,i})_{t,i}$ are fixed in the VMOT problem, all functions appearing in the theorem can be regarded as essentially real-valued measurable functions, with no further regularity imposed.

\subsection{Intuitive Explanation and Approximating Dual Maximizers}

Upgrading the duality relation \eqref{duality} to full dual attainment requires a refined approximation argument. The key idea, originating in \cite{bj,bnt} for the single-asset, two-period case, is to construct a sequence of ``nice'' dual candidates and then extract a pointwise limit using local compactness and irreducibility.

We call a sequence $(\phi_n, h_n)_{n \in \mathbb{N}}$ an \emph{approximating dual maximizer} if, for each $n$:
\begin{itemize}
\item $\phi_{t,i,n}$ is real-valued, continuous, and integrable with respect to $\mu_{t,i}$, i.e., $\phi_{t,i,n} \in L^1(\mu_{t,i})$;
\item $h_{t,i,n}$ is bounded and continuous (with the convention $h_{N,i,n} \equiv 0$);
\item the following pathwise inequality holds:
\begin{align}
\label{dual}
\sum_{t=1}^N \sum_{i=1}^d \phi_{t,i,n}(x_{t,i})
+ h_{t,i,n}(\bar x_t)\,\Delta x_{t,i}
\;\le\; c(x)
\quad \text{for all } x \in \R^{Nd},
\end{align}
where $\bar x_t := (x_1,\dots,x_t)\in\R^{dt}$ and $\Delta x_{t,i} := x_{t+1,i} - x_{t,i}$;
\item and the dual values converge increasingly to $D(c)$:
\begin{align}\label{maximizing}
\mu(\phi_n)
:= \sum_{t=1}^N \sum_{i=1}^d \int \phi_{t,i,n}(x_{t,i})\,d\mu_{t,i}(x_{t,i})
\nearrow D(c)
\quad \text{as } n \to \infty.
\end{align}
\end{itemize}

Adapting the strategy of \cite{bj,bnt} to the multi-asset, multi-period setting, the next proposition provides pointwise convergence of a subsequence of $(\phi_n,h_n)$, which is a crucial step toward Theorem~\ref{main}.

\begin{proposition}\label{ptwiseconverge}
Under the assumptions of Theorem~\ref{main}, there exists an approximating dual maximizer $(\phi_n, h_n)_{n \in \mathbb{N}}$ such that, for each $t \in [N]$ and $i \in [d]$, the sequence $\phi_{t,i,n}$ converges $\mu_{t,i}$-a.s.\ to a real-valued function $\phi_{t,i}$ as $n \to \infty$.
\end{proposition}

The proof of Proposition~\ref{ptwiseconverge} relies on a structural property of the marginals, namely irreducibility in convex order, together with a local $L^1$ bound. We briefly recall the relevant notion (see \cite{bnt} for details). A pair of probability measures $\mu \preceq_c \nu$ (with finite first moments) on $\R$ is called \emph{irreducible} if the open set
\[
I := \{x \in \R : u_{\mu}(x) < u_{\nu}(x)\}
\]
is a connected interval and $\mu(I) = \mu(\R)$, where
\[
u_{\mu}(x) := \int_\R |x-y|\,d\mu(y)
\]
is the potential function of $\mu$. In this case, we define the \emph{domain} $(I,J)$ of $(\mu,\nu)$ by letting $J$ be the smallest interval with $\nu(J) = \nu(\R)$. Then $J$ is the union of $I$ and any endpoints of $I$ that are atoms of $\nu$, and one has $I = \mathrm{int}(J)$. The interval $J$ may be bounded or unbounded, and may take the form $(a,b]$, $[a,b)$, $(a,b)$, or $[a,b]$, with $I=(a,b)$ in all cases. Intuitively, irreducibility expresses that $\nu$ is ``spread out'' from $\mu$ in a regular way. It is a natural and generic property: most pairs of distributions in convex order satisfy irreducibility, and pairs that do not can typically be perturbed slightly to obtain it.

Let $(I_{t,i},J_{t,i})$ denote the domain of the irreducible pair $\mu_{t,i} \preceq_c \mu_{t+1,i}$ for $t\in[N-1]$, with the conventions $J_{0,i}:=J_{1,i}$ and $I_{N,i}:=J_{N-1,i}$. We have $J_{t,i}\subset I_{t+1,i}$ for all $t \in [N-1]$ and $i\in[d]$. The next lemma provides a local $L^1$ bound on the approximating potentials.

\begin{lemma}\label{L1bound*}
Under the assumptions of Theorem~\ref{main}, there exists an approximating dual maximizer $(\phi_n, h_n)_{n \in \mathbb{N}}$. Moreover, for each $t \in [N-1]$ and $i \in [d]$, there exists an increasing sequence of compact intervals $(J_{t,i,k})_{k \in \mathbb{N}}$ such that $\bigcup_{k \in \mathbb{N}} J_{t,i,k} = J_{t,i}$ and
\begin{align}\label{supbound*}
\sup_{n} 
\left\|
\phi_{t,i,n} - \int \phi_{t,i,n}\,d\mu_{t,i,k}
\right\|_{L^1(\mu_{t,i,k})}
\;\le\; C_k,
\end{align}
where $\mu_{t,i,k} := \frac{1}{\mu_{t,i}(J_{t-1,i,k})}\,\mu_{t,i}\big|_{J_{t-1,i,k}}$ is the normalized restriction of $\mu_{t,i}$ to $J_{t-1,i,k}$, and $C_k$ depends only on $k$ (not on $n$).
\end{lemma}

The local $L^1$ bound in Lemma~\ref{L1bound*} is a key ingredient for extracting pointwise convergent subsequences and thus for proving Proposition~\ref{ptwiseconverge} and Theorem~\ref{main}. Beyond its technical role, it is also of independent interest in the analysis of martingale optimal transport with general cost functions. Detailed proofs of Proposition~\ref{ptwiseconverge} and Lemma~\ref{L1bound*} are given in Section~\ref{sec:proofs}.

These duality and dual attainment results establish a rigorous foundation for robust pricing and hedging in multi-asset, multi-period markets. In the following sections, we show how they translate into practical tools for model calibration, risk management, and the pricing of complex path-dependent derivatives.

\section{Numerical Methods and Implementation}
\label{sec:numerical-methods}

As discussed in \Cref{sec:financial-applications}, the martingale optimal transport (MOT) framework has emerged as a powerful tool in quantitative finance, enabling both model calibration to observed market marginals and robust pricing and hedging of complex, multi-asset derivatives. 
In particular, MOT offers a non-parametric, model-independent methodology that naturally accommodates path-dependence and requires only minimal assumptions on asset dynamics. Moreover, it allows practitioners to incorporate additional market information—such as intermediate marginals or correlation-sensitive products—to further tighten price bounds.
These theoretical advances, particularly in the multi-marginal and multi-asset setting, underscore the importance of scalable and accurate numerical methods. Motivated by these applications, we now turn to the numerical implementation that make the solution of high-dimensional MOT problems feasible in practice.

In this section, we detail the computational strategies employed to solve the multi-marginal MOT problem, emphasizing the challenges posed by high-dimensionality and path-dependent constraints. We introduce the primal-dual linear programming (PDLP) approach, describe its algorithmic enhancements, and present the discrete formulation of the MOT problem suitable for large-scale numerical resolution.

\subsection{Computational Challenges}

The numerical resolution of the martingale optimal transport (MOT) problem has garnered considerable attention in recent literature. 
In the multi-marginal case, the direct solution of the underlying linear programming (LP) formulation of MOT suffers from the curse of dimensionality, as the computational complexity grows exponentially with the number of time periods and the number of dimensions of the process. 

Conventional methods for solving MOT include entropy-regularized approaches, most notably the Sinkhorn algorithm, and neural network (NN) based methods. The Sinkhorn algorithm has been successfully applied to the primal formulation of optimal transport problems by exploiting Bregman projections to solve the regularized linear program~\cite{cuturi2013sinkhorn}. 
De March~\cite{demarch2018entropic} proposes a modified Sinkhorn algorithm for bimarginal MOT, which replaces the standard Bregman projections with an approximate projection procedure that iteratively enforces the martingale constraints through a fixed-point update mechanism, thereby achieving adequate accuracy while mitigating the numerical instabilities inherent in the exact projection computations.
However, significant challenges arise when extending this approach to the multi-marginal MOT framework. 
Moreover, the Bregman projections essential to the Sinkhorn algorithm are problematic to compute under the martingale constraints; indeed, even in the bimarginal case, these projections are solved only approximately, which compromises accuracy. Additionally, the iterative nature of Sinkhorn renders it susceptible to significant numerical instability when handling the intricate constraints imposed by the martingale requirement.

Neural network methods, which have been developed to tackle the dual formulation of the MOT problem~\cite{eckstein2021robust}, offer a flexible approach to handle high-dimensional data. However, these methods currently do not offer theoretical guarantees regarding convergence or precision, which limits their reliability in scenarios where rigorous accuracy is paramount. 

Given these challenges, there is a need for scalable, robust algorithms that can efficiently handle the high-dimensional, sparse and path-dependent structure of multi-marginal MOT problems.

\subsection{PDLP Algorithm and Enhancements}
To address these computational challenges, we employ the primal-dual linear programming (PDLP) framework~\cite{applegate2021practical,hinder2023worst}, which is well-suited for large-scale linear programs. 
PDLP is a first-order method that combines a saddle-point reformulation of the LP with adaptive algorithmic enhancements. 
It reformulates the standard LP as a saddle-point problem and iteratively updates primal and dual variables using matrix-vector products, avoiding explicit matrix storage and factorization.
Its inherent matrix-free nature and reliance solely on matrix-vector products make PDLP particularly well-suited to the high-dimensional and sparse settings often encountered in MOT problems.

In brief, PDLP begins by rewriting the standard LP formulation
\begin{equation}
  \begin{aligned}
    \min_{x \in \mathbb{R}^n} \; & c^\top x, \\
    \text{s.t.}\; & Ax = b,\quad x \ge 0,
  \end{aligned}
  \label{eq:lp_formulation}
\end{equation}
into an equivalent saddle-point (or primal-dual) problem:
\begin{equation}
  \min_{x \in X}\max_{y \in Y} \; L(x,y) \quad \text{with} \quad L(x,y) = c^\top x - y^\top Ax + b^\top y,
  \label{eq:saddle_point_lp}
\end{equation}
where $X=\mathbb{R}_{\ge 0}^n$ and $Y=\mathbb{R}^m$. 
The \emph{primal-dual hybrid gradient} (PDHG) scheme~\cite{chambolle2011first} is then applied by iteratively performing the updates
\begin{align}
  x_{k+1} &= \operatorname{proj}_{X} \left(x_k - \tau \left(c - A^\top y_k\right)\right), \label{eq:pdhg_primal} \\
  y_{k+1} &= y_k + \sigma \left(b - A\left(2x_{k+1} - x_k\right)\right), \label{eq:pdhg_dual}
\end{align}
where $\tau$ and $\sigma$ denote the primal and dual step sizes, respectively.

PDLP enhances the baseline PDHG method by incorporating several key features:
\begin{itemize}
  \item \textbf{Adaptive Step Size.} An adaptive heuristic adjusts the step size dynamically to guarantee that the step condition
    \[
      \tau, \sigma \leq \frac{\|z_{k+1}-z_k\|_\omega^2}{2\,(y_{k+1}-y_k)^\top A(x_{k+1}-x_k)}
    \]
    is satisfied. This avoids overly conservative fixed-step estimates (e.g., $\tau = \sigma = 1/\|A\|^2$) and accelerates convergence.
  \item \textbf{Adaptive Restarts.} Periodic restarts based on the normalized duality gap (which remains finite even when the standard duality gap is unbounded) ensure that progress is regularly “reset” to counteract self-inhibiting tailing-off, thereby leading to empirical linear convergence.
  \item \textbf{Primal Weight Updates.} During restarts, a weight parameter is updated to balance the progress in the primal and dual variables. In effect, the update aims to equalize the weighted distances to optimality in both spaces.
  \item \textbf{Presolving and Diagonal Preconditioning.} Prior to the main iterations, problem data are simplified using presolve routines and then rescaled via a diagonal preconditioner. Such scaling helps to alleviate numerical imbalances, thereby reducing the effective condition number of the data matrix.
\end{itemize}

From a computational complexity perspective, recent analyses have established rigorous bounds for the PDLP algorithm in both special and general cases. While PDLP achieves particularly strong complexity guarantees for totally unimodular constraint matrices~\cite{hinder2023worst}, it is important to note that total unimodularity of the constraint matrix $A$ is a property specific to the classical bimarginal optimal transport (OT) problem. In contrast, the constraint matrices encountered in multimarginal OT and in bi- and multi-marginal martingale optimal transport (MOT) problems generally lack total unimodularity due to the more complex structure of the marginal and martingale constraints.

For general linear programs, PDLP requires at most
$O\left(\frac{L^2 R^2}{\epsilon^2}\right)$
matrix-vector multiplications to reach an $\epsilon$-optimal solution~\cite{applegate2021practical}, where $L$ is a Lipschitz constant related to the problem data, $R$ is a bound on the feasible region, and $\epsilon$ is the desired accuracy. Although this worst-case bound is less favorable than in the unimodular case, PDLP remains highly effective in practice for large, sparse, and structured problems such as those encountered in multi-marginal MOT. Its matrix-free implementation and scalability make it a robust choice for high-dimensional optimal transport computations.

This general complexity result underscores the practical viability of PDLP for MOT applications, even when the underlying linear program does not possess special structure. The algorithm's ability to exploit sparsity and avoid explicit matrix factorizations is particularly advantageous in the high-dimensional, path-dependent setting of multi-marginal MOT.

\subsection{Discrete Multimarginal MOT Formulation}
For numerical implementation, we discretize the multi-marginal MOT problem over a finite grid of asset prices and time steps.
Specifically, we consider a MOT problem that spans multiple assets and time steps, defined over a discrete support. 
The discrete formulation is given by the following linear program:
\begin{problem}[discrete multimarginal MOT]
\label{prob:multi-mot}
\begin{equation}\tag{multiMOT}\label{eq:multi-mot}
    \begin{aligned}
        \min_{\pi \geq 0} \quad & \langle C , \pi \rangle \\
        \text{subject to} \quad & \proj_{(t,k)} \left( \pi \right) = \mu_{t,k}, \quad \forall t \in [T], \ k \in [d], \\
        & \mathbb{E}_{\pi} \left[s_{t+1,k} - s_{t,k} \mid \mathcal{F}_t \right]= 0, \quad  \forall t \in [T-1], \ k \in [d], 
    \end{aligned}
\end{equation}
where $C = [C_{i_{1,1}, \dots i_{T,d}}]_{i_{1,1} \in [n_{1,1}], \dots, i_{T,d} \in [n_{T,d}]}$ denotes the non-negative cost tensor with rank $m = Td$, $\mu_{t,k} = [\mu_{t,k,i}]_{i \in [n_{t,k}]}$ is the probability vector representing the marginal of the $k$-th process at the $t$-th time step.
\end{problem}

To account for discretization error, we relax the martingale constraints by introducing a tolerance parameter corresponding to the grid size:
\begin{problem}[discrete multimarginal MOT with relaxed martingale constraints]
\label{prob:multi-mot-relaxed}
\begin{equation}\tag{multiMOTrelaxed}\label{eq:multi-mot-relaxed}
    \begin{aligned}
        \min_{\pi \geq 0} \quad & \langle C , \pi \rangle \\
        \text{subject to} \quad & \proj_{(t,k)} \left( \pi \right) = \mu_{t,k}, \quad \forall t \in [T], \ k \in [d], \\
        & \mathbb{E}_{\pi} \left[s_{t+1,k} - s_{t,k} \mid \mathcal{F}_t \right] \le \frac{\Delta_{t,k}}{2}, \quad  \forall t \in [T-1], \ k \in [d], \\
        & \mathbb{E}_{\pi} \left[s_{t+1,k} - s_{t,k} \mid \mathcal{F}_t \right] \ge -\frac{\Delta_{t,k}}{2}, \quad  \forall t \in [T-1], \ k \in [d], 
    \end{aligned}
\end{equation}
where $C = [C_{i_{1,1}, \dots i_{T,d}}]_{i_{1,1} \in [n_{1,1}], \dots, i_{T,d} \in [n_{T,d}]}$ denotes the non-negative cost tensor with rank $m = Td$, $\mu_{t,k} = [\mu_{t,k,i}]_{i \in [n_{t,k}]}$ is the probability vector representing the marginal of the $k$-th process at the $t$-th time step.
\end{problem}
This discrete formulation enables the use of efficient large-scale linear programming solvers, such as the PDLP implementation in NVIDIA cuOpt, to compute optimal transport plans and corresponding dual certificates under realistic market constraints. Leveraging this computational framework, we now present the results of our numerical experiments, including solution quality metrics, illustrative case studies, and an analysis of both primal and dual solutions.

\section{Numerical Experiments and Results}
\label{sec:numerical-results}

In this section, we present the results of our numerical experiments for the MOT problem. We begin by describing the metrics used to evaluate solution quality, detail the construction of marginal data from market information, analyze the computed primal and dual solutions, and provide a case study on a worst-of autocallable option. We conclude with a discussion of the implications and limitations of our findings.

\subsection{Metrics for solution quality}

To rigorously assess the quality and reliability of the computed solutions to the MOT problem, we employ a suite of quantitative metrics that capture both optimality and feasibility aspects. These metrics are designed to provide a comprehensive evaluation of the numerical performance of the PDLP approach under the high-dimensional and path-dependent constraints characteristic of the MOT framework.

The primary metric is the \emph{primal objective value}, which represents the expected cost (or payoff) under the optimal transport plan for the given cost tensor. Complementing this, the \emph{dual objective value} aggregates the contributions from the dual variables associated with the marginal constraints, weighted by the prescribed marginals. The difference between these two values, known as the \emph{duality gap}, serves as a certificate of optimality: a small gap indicates that the computed solution is close to optimal, while a larger gap may signal numerical issues or suboptimal convergence.

Feasibility with respect to the problem constraints is evaluated through measures of \emph{primal infeasibility} and \emph{dual infeasibility}. Primal infeasibility quantifies the deviation of the computed transport plan from the prescribed marginal distributions, while dual infeasibility measures the extent to which the dual variables fail to satisfy the dual constraints imposed by the cost structure and the martingale conditions.

To summarize these deviations, we report the $\ell_{p}$-norms of the infeasibility vectors, including $\ell_{1}$, $\ell_{2}$, and $\ell_{\infty}$ norms. These norms provide insight into the aggregate, average, and worst-case violations, respectively, and are useful for benchmarking solver performance across different problem instances and parameter regimes.

For clarity, the key metrics are:
\begin{itemize} 
    \item \textbf{Primal objective value} ($V_{p}$): Expected cost or payoff under the optimal transport plan. 
    \item \textbf{Dual objective value} ($V_{d}$): Aggregate value from dual variables, weighted by the prescribed marginals. 
    \item \textbf{Duality gap} ($G$): Difference between primal and dual objective values, indicating solution optimality. \item \textbf{Primal infeasibility} ($\delta^{p}$): Deviation of the computed transport plan from the prescribed marginal distributions. 
    \item \textbf{Dual infeasibility} ($\delta^{d}$): Extent to which dual variables fail to satisfy dual constraints. 
    \item \textbf{$\ell_{p}$-norms of infeasibility}: Aggregate ($\ell_{1}$), average ($\ell_{2}$), and worst-case ($\ell_{\infty}$) violations of constraints. 
\end{itemize}
Together, these metrics enable a transparent and robust evaluation of both the accuracy and feasibility of the computed solutions, facilitating meaningful comparisons with alternative methods and guiding further algorithmic improvements.

\subsection{Case study: worst-of autocallable option}

To showcase the practical relevance of our approach, we consider a worst-of autocallable option referencing the S\&P 500 and NASDAQ 100 indices. Autocallable options are a prominent class of structured products in finance, widely traded in equity-linked markets due to their attractive risk-return profiles and embedded path-dependent features. These instruments offer periodic opportunities for early redemption (autocall), contingent on the performance of one or more underlying assets relative to specified barriers. Their popularity stems from the ability to generate enhanced yields in low-volatility environments, provide partial downside protection, and facilitate tailored risk exposures for both retail and institutional investors. The liquidity and market depth of autocallable products, especially those linked to major equity indices, make them a natural choice for robust pricing and risk management studies. Moreover, the worst-of variant, where the payoff depends on the least-performing asset in a basket, is particularly relevant for risk transfer and capital protection strategies.

\begin{definition}[Multi-Asset Autocall Payoff]
\label{def:autocall-payoff}
Let $\bm{S}(t) := (S_1(t), S_2(t), \ldots, S_d(t))$ for $t \in [0, T]$ denote the vector of stochastic processes representing the prices of $d$ underlying assets.  
Consider a sequence of observation dates $0 < t_1 < t_2 < \cdots < t_m = T$, where barrier and knock-out conditions are evaluated and coupons are accrued.  
Let $K$ be the strike price, $B_{\mathrm{KI}}$ the knock-in barrier, $B_{\mathrm{KO}}$ the knock-out barrier, and $C_j$ the coupon accrued over $[t_{j-1}, t_j]$ for $j \in [m]$ (with $t_0 := 0$).

Define the worst-performing asset at time $t$ as
\begin{align*}
    X(t) := \min_{i \in [d]} S_i(t).
\end{align*}

Define the knock-out stopping time as
\begin{align*}
    \tau := \inf \left\{ j \in [m] : X(t_j) \geq B_{\mathrm{KO}} \right\},
\end{align*}
with the convention that $\inf \varnothing = +\infty$.

The payoff at time $t_k$ for $k \in [m]$ of a worst-of multi-asset autocall is given by
\begin{align*}
    f_k\left( \bm{S}(t_1), \ldots, \bm{S}(t_k) \right) =
    \begin{cases}
        \sum_{j=1}^{\tau} C_j, & \text{if } k = \tau < m, \\[1.5ex]
        \min \{ X(T) - K, 0 \}, & \text{if } k = m \leq \tau \text{ and } X(T) \leq B_{\mathrm{KI}}, \\[1.5ex]
        \sum_{j=1}^{m} C_j, & \text{if } k = m \leq \tau \text{ and } X(T) > B_{\mathrm{KI}}, \\[1.5ex]
        0, & \text{otherwise, i.e., } k < \min\{m, \tau\}.
    \end{cases}
\end{align*}
\end{definition}

In our specific application, we consider a worst-of autocallable structure linked to two major equity indices: the S\&P 500 (SPX) and the NASDAQ 100 (NDX), corresponding to $d=2$ underlying assets. The product features a three-year maturity, with an inception date $t_{\mathrm{inc}} = -\frac{5}{24}$ years and three annual monitoring dates: $t_1 = t_{\mathrm{inc}} + 1$, $t_2 = t_{\mathrm{inc}} + 2$, and $t_3 = t_{\mathrm{inc}} + 3$ years. The knock-out barrier is set at $B_{\mathrm{KO}} = 120\%$ of the initial level, and the knock-in barrier at $B_{\mathrm{KI}} = 60\%$. The strike price is $K = 100\%$. The coupon is a fixed rate, with an annualized rate of $8\%$, accrued at each period, i.e., $C_j = 8\% \times (t_j - t_{j-1})$ for $j \in \{1,2,3\}$. We perform present value calculations as of $t_0 = 0$, assuming a constant discount rate of $1\%$.\footnote{In practice, the \emph{inception date} refers to the date when the product is issued and the initial asset levels are fixed, while the \emph{evaluation date} is the date at which the product is being valued or analyzed, typically for risk management or pricing purposes. The separation allows for the modeling of products that have already been running for some time, and for the use of current market data in valuation.} 

It is important to note that the payoff function described above is defined \emph{per unit notional}. In practical applications, the final payout to the investor is obtained by multiplying the computed payoff by the notional amount of the contract.

At each monitoring date, the reference level is determined by the worst-performing underlying between SPX and NDX. For example, suppose at $t_1 = t_{\mathrm{inc}} + 1$ years, SPX is at $125\%$ and NDX is at $130\%$ of their respective initial levels. The worst-performing asset is SPX, which exceeds the $120\%$ knock-out barrier, triggering the autocall feature. The product is then redeemed early, and the holder receives the notional plus the accrued coupon up to that date.

If neither asset breaches the knock-out barrier at any monitoring date, the product continues to maturity. At $t_3 = t_1 + 3$ years, if the worst-performing asset is below the knock-in barrier ($60\%$), the holder receives a down-and-in put payoff, $\min\{X(T) - K, 0\}$, reflecting the loss relative to the strike. If the worst-performing asset remains above the knock-in barrier, the holder receives the full coupon accrued over the three years.

This structure provides a realistic and challenging test case for the martingale optimal transport framework, as the payoff is highly path-dependent and sensitive to joint asset dynamics. The dimensionality of the problem is $dm = 2 \times 3 = 6$ for the two assets and three monitoring dates, making it tractable for numerical solution via large-scale linear programming. In the following sections, we present numerical results illustrating dual attainment and optimal transport plans for this worst-of autocallable payoff, computed using PDLP solvers under the prescribed marginal and martingale constraints.

\subsection{Marginal data construction}

To construct the marginal distributions required for the multi-period, multi-asset martingale optimal transport (MOT) problem, we utilize market data in the form of vanilla option prices for the S\&P 500 (SPX) and NASDAQ 100 (NDX) indices, evaluated as of December 5, 2025 ($t_0$). These option prices are used to infer the risk-neutral marginal distributions of the underlying asset prices at the three monitoring dates specified in the autocallable payoff definition (see~\Cref{def:autocall-payoff}).

For each asset and each monitoring date, the marginal distribution is discretized over a support of $n=10$ points, chosen to adequately capture the range of plausible price scenarios implied by the market. The support points are scaled by the respective spot prices at the evaluation date ($t_0$), ensuring comparability and normalization across assets and time steps.

The construction of these marginals from market data follows a standard procedure in quantitative finance. First, a smooth implied volatility surface is calibrated to the observed vanilla option prices for each asset and maturity. Then, using the Breeden--Litzenberger formula~\cite{bl78}, the risk-neutral probability density function is extracted by taking the second derivative of the call price with respect to strike. The resulting densities are discretized to form the input marginals for the MOT problem.

The resulting discrete marginal distributions for SPX and NDX at all three monitoring dates are visualized in~\Cref{fig:marginal-distributions-2d}. These marginals serve as input constraints for the MOT problem, enabling robust pricing and hedging of the worst-of autocallable payoff under the prescribed market information.

\begin{figure}[t]
    \centering
    \includegraphics[width=1.0\linewidth]{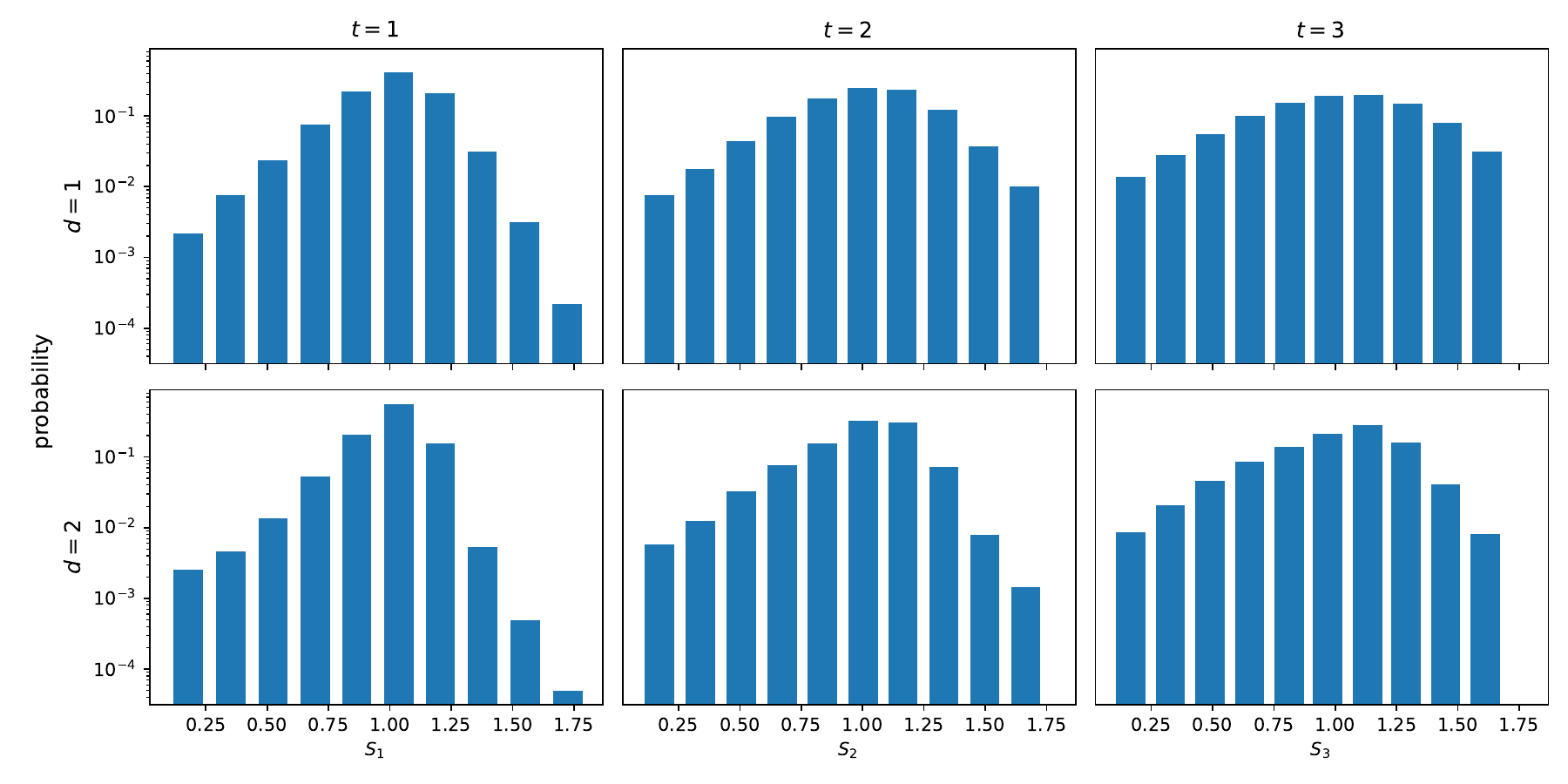}
    \caption{Discrete marginal distributions used as input to the $2$-dimensional, $3$-time-step MOT problem. Asset prices are scaled by their respective spot prices at $t_0$.}
    \label{fig:marginal-distributions-2d}
\end{figure}

\subsection{Discussion of results}

In our numerical experiments, we employ the PDLP solver implemented within NVIDIA cuOpt~\cite{nvidia_cuopt}, running on GPU hardware, to solve the discrete multi-marginal MOT problems. Leveraging GPU acceleration enables efficient large-scale computation and significantly reduces solution times for high-dimensional instances. The solver is configured with an optimality tolerance of $10^{-12}$, ensuring that the primal and dual infeasibilities, as well as the optimality gap, are all maintained below this threshold in both absolute and relative terms. This stringent tolerance guarantees that the computed solutions are not only highly accurate but also robust with respect to the complex constraints inherent in the MOT framework.

\begin{figure}[h!]
    \centering
    \begin{subfigure}{0.48\textwidth}
        \includegraphics[width=\textwidth]{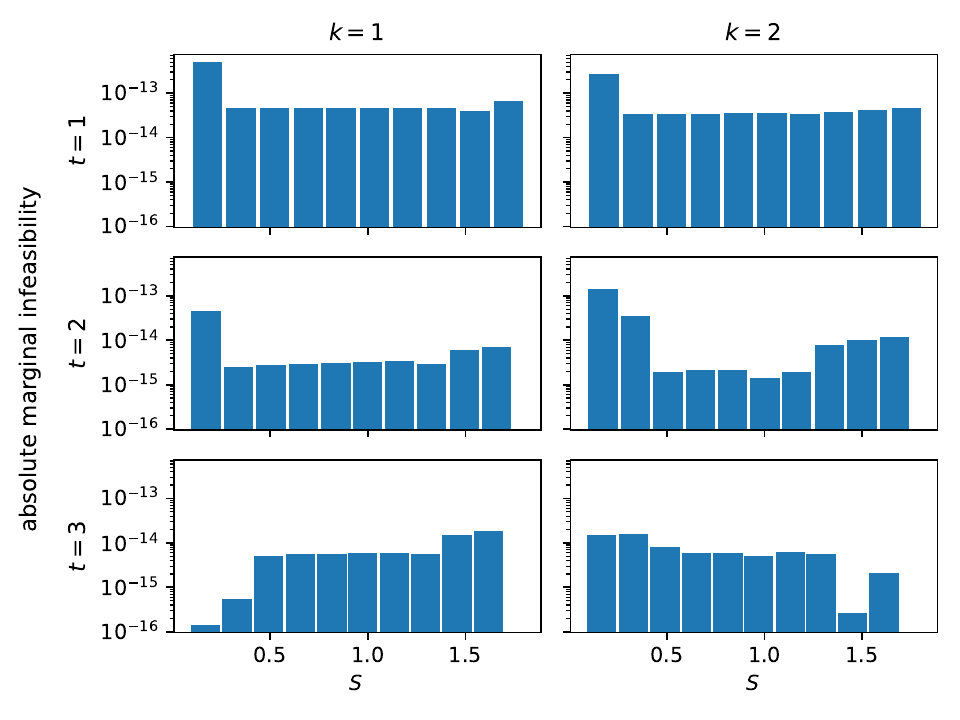}
        \caption{primal maximization}
        \label{fig:2d-marginal-infeas-max}
    \end{subfigure}
    \hfill
    \begin{subfigure}{0.48\textwidth}
        \includegraphics[width=\textwidth]{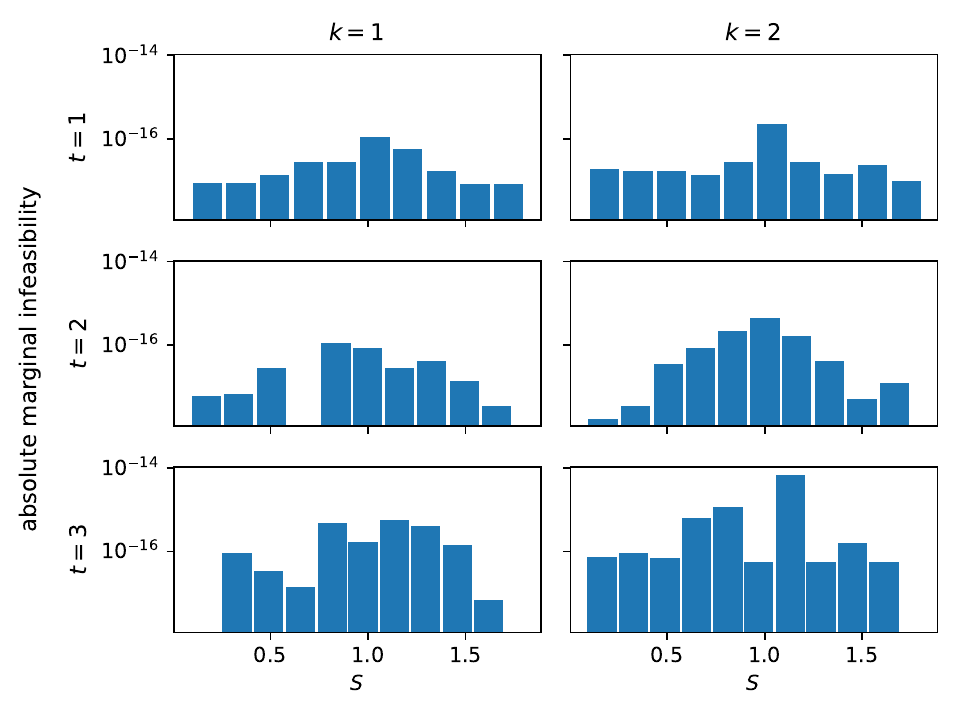}
        \caption{primal minimization}
        \label{fig:2d-marginal-infeas-min}
    \end{subfigure}
    \caption{Absolute infeasibility with respect to the marginal constraints in the primal solution of the $2$-dimensional, $3$-time-step MOT problem. Results are shown for both maximization and minimization directions, computed using the PDLP solver in NVIDIA cuOpt with an optimality tolerance of $10^{-12}$. The low levels of infeasibility confirm the solver's ability to enforce marginal constraints with high precision.}
    \label{fig:2d-marginal-infeas}
\end{figure}
\begin{figure}[h!]
    \centering
    \begin{subfigure}{0.48\textwidth}
        \includegraphics[width=\textwidth]{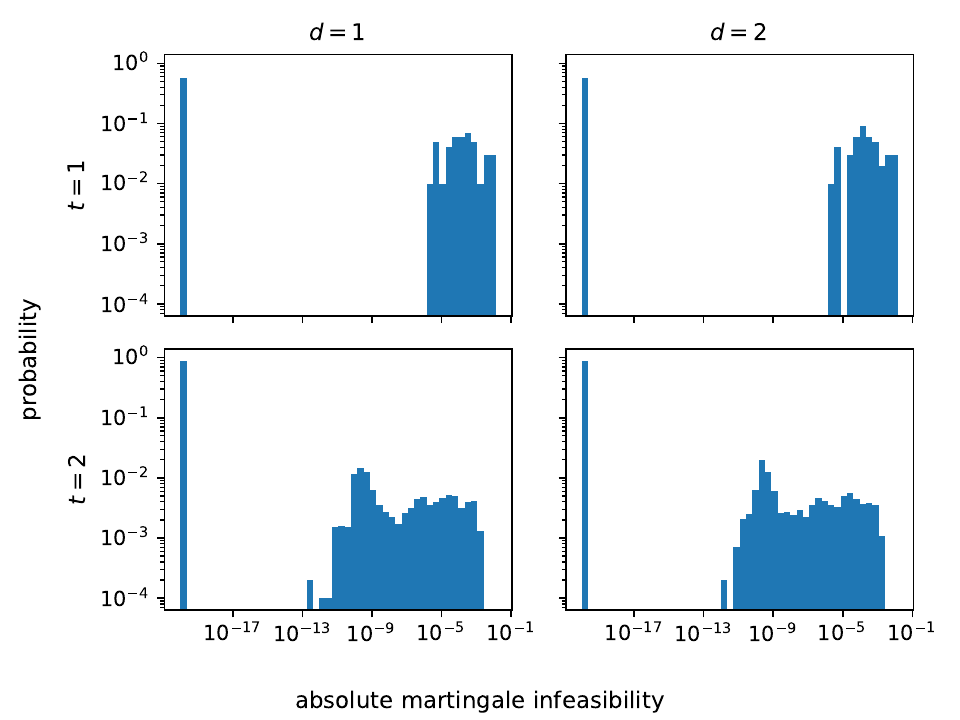}
        \caption{primal maximization}
        \label{fig:2d-martingal-infeas-max}
    \end{subfigure}
    \hfill
    \begin{subfigure}{0.48\textwidth}
        \includegraphics[width=\textwidth]{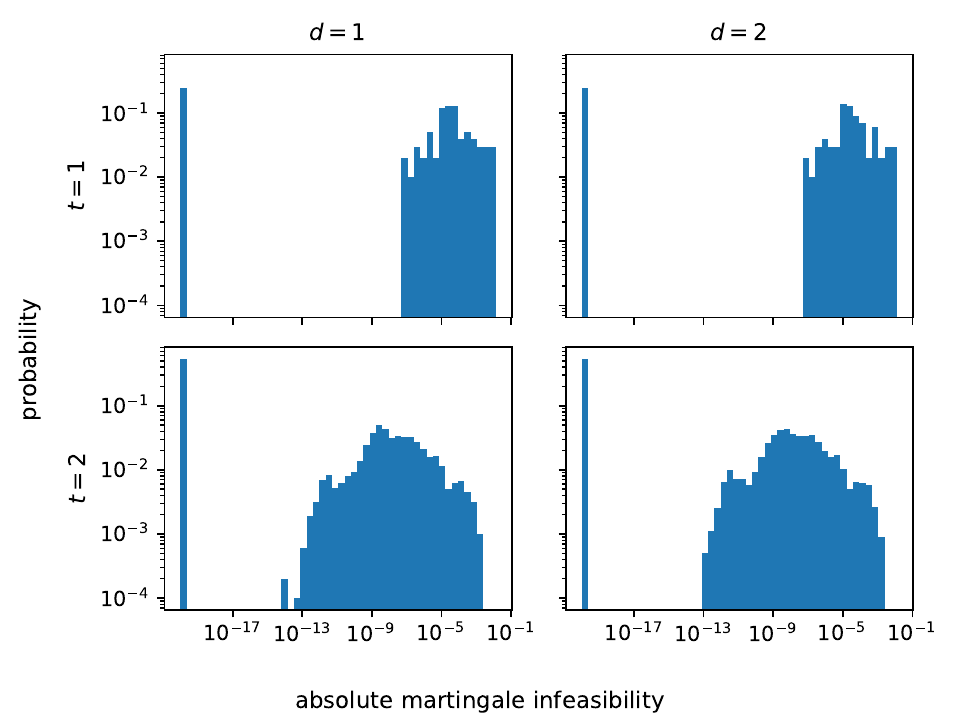}
        \caption{primal minimization}
        \label{fig:2d-martingal-infeas-min}
    \end{subfigure}
    \caption{Absolute infeasibility with respect to the martingale constraints in the primal solution of the $2$-dimensional, $3$-time-step MOT problem. The PDLP solver in cuOpt maintains infeasibility below the grid discretization error for both optimization directions, demonstrating robust enforcement of path-dependent martingale constraints.}
    \label{fig:2d-martingal-infeas}
\end{figure}

The primal solution, represented by the optimal transport plan $\pi$, encodes the joint probability distribution over asset prices and time steps that minimizes (or maximizes) the expected payoff under the specified cost function, subject to the marginal and martingale constraints. \Cref{fig:2d-marginal-infeas,fig:2d-martingal-infeas} illustrate the infeasibility of the primal solution with respect to the marginal and martingale constraints, respectively, for both the maximization and minimization problems. The results demonstrate that the PDLP solver achieves high accuracy, with infeasibility levels well within the prescribed tolerances. In particular, the relaxed martingale constraints are satisfied within the grid discretization error, confirming the effectiveness of the presolve and preconditioning steps in the numerical pipeline.

\begin{figure}[h!]
    \centering
    \begin{subfigure}{0.48\textwidth}
        \includegraphics[width=\textwidth]{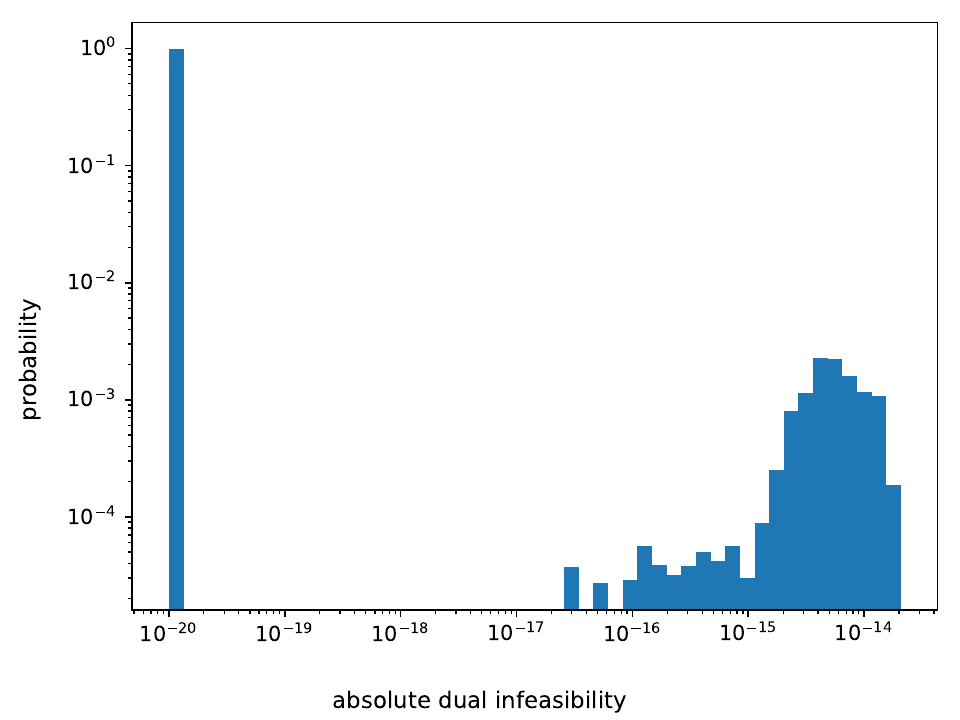}
        \caption{primal maximization}
        \label{fig:2d-dual-infeas-max}
    \end{subfigure}
    \hfill
    \begin{subfigure}{0.48\textwidth}
        \includegraphics[width=\textwidth]{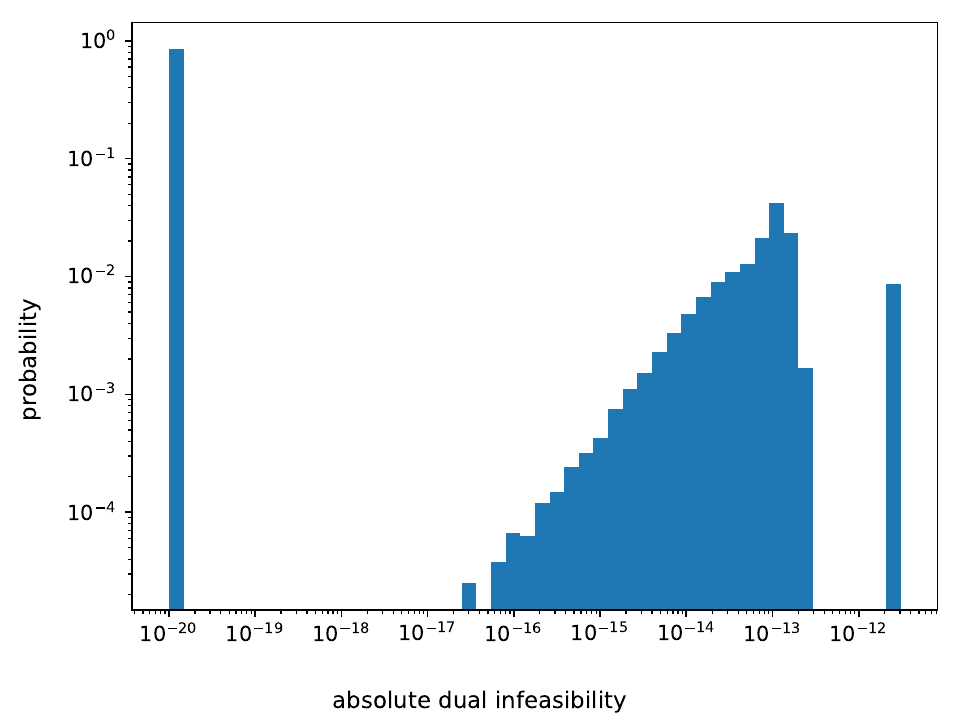}
        \caption{primal minimization}
        \label{fig:2d-dual-infeas-min}
    \end{subfigure}
    \caption{Absolute infeasibility with respect to the dual constraints in the primal solution of the $2$-dimensional, $3$-time-step MOT problem. The results highlight the tightness of the dual bounds achieved by the PDLP solver in cuOpt, with infeasibility consistently below the prescribed tolerance.}
    \label{fig:2d-dual-infeas}
\end{figure}

\begin{table}[h!]
\centering
\begin{tabular}{|ll|l|l|}
\hline
\multicolumn{2}{|c|}{\textbf{Optimization   Direction}}                                                                   & \textbf{Maximization} & \textbf{Minimization} \\ \hline
\multicolumn{2}{|c|}{\textbf{Optimal Primal   Objective Value}}                                                           & $0.187880$            & $0.062208$           \\ \hline
\multicolumn{2}{|c|}{\textbf{Optimal Dual   Objective Value}}                                                             & $0.187880$            & $0.062208$           \\ \hline
\multicolumn{2}{|c|}{\textbf{Duality Gap}}                                                                                & $-7.7577 \times 10^{-14}$ & $-1.1971 \times 10^{-13}$ \\ \hline
\multicolumn{1}{|c|}{\multirow{3}{*}{\textbf{Primal Infeasibility}}}                                         & \textbf{$\ell_{1}$}      & $4.1653 \times 10^{-12}$  & $1.3085 \times 10^{-14}$  \\ \cline{2-4} 
\multicolumn{1}{|c|}{}                                                                                                  & \textbf{$\ell_{2}$}      & $1.6503 \times 10^{-12}$  & $7.0355 \times 10^{-15}$  \\ \cline{2-4} 
\multicolumn{1}{|c|}{}                                                                                                  & \textbf{$\ell_{\infty}$} & $1.4217 \times 10^{-12}$  & $6.8279 \times 10^{-15}$  \\ \hline
\multicolumn{1}{|c|}{\multirow{3}{*}{\textbf{Dual Infeasibility}}}                                           & \textbf{$\ell_{1}$}      & $7.0409 \times 10^{-11}$  & $3.7987 \times 10^{-8}$  \\ \cline{2-4} 
\multicolumn{1}{|c|}{}                                                                                                  & \textbf{$\ell_{2}$}      & $7.7958 \times 10^{-13}$  & $2.8302 \times 10^{-10}$  \\ \cline{2-4} 
\multicolumn{1}{|c|}{}                                                                                                  & \textbf{$\ell_{\infty}$} & $2.0761 \times 10^{-14}$  & $3.0800 \times 10^{-12}$  \\ \hline
\end{tabular}
\caption{Summary of numerical results for the $2$-dimensional, $3$-time-step MOT problem solved using the PDLP solver in NVIDIA cuOpt. The table reports optimal primal and dual objective values, duality gaps, and infeasibility norms for both maximization and minimization. All metrics are on the order of or below $10^{-8}$, with the majority well below $10^{-10}$, confirming the accuracy and robustness of the computed solutions.}
\label{tab:mot-2d-numerical-results}
\end{table}

The dual solution, comprising the variables $\phi$ and $h$, provides valuable insight into the shadow prices associated with the marginal and martingale constraints. These dual variables can be interpreted as sensitivities of the optimal value to perturbations in the input marginals and the martingale conditions, respectively. \Cref{fig:2d-dual-infeas} presents the infeasibility of the dual solution, further corroborating the tightness of the computed bounds and the robustness of the solver.

\Cref{tab:mot-2d-numerical-results} summarizes the key numerical results, including the optimal primal and dual objective values, duality gaps, and infeasibility norms for both maximization and minimization directions. The duality gaps are on the order of $10^{-13}$ or smaller, demonstrating that the solutions achieve optimality to within nearly machine precision. Primal and dual infeasibility norms are similarly small, reflecting the solver's ability to enforce the complex constraints of the MOT problem in a high-dimensional, path-dependent setting.

Overall, the computed primal and dual solutions validate the theoretical guarantees of the PDLP approach and demonstrate its practical effectiveness for large-scale, multi-asset, multi-period MOT problems. The use of GPU-accelerated computation via NVIDIA cuOpt was essential for handling the high-dimensional, path-dependent structure of these problems, enabling robust and scalable solution of MOT instances that would be infeasible on standard CPU-based solvers. These results provide a foundation for further applications in robust pricing, risk management, and model calibration in financial engineering and related fields.

\section{Conclusion and Future Work}
\label{sec:conclusion}

In this work, we have developed a comprehensive framework for dual attainment in the multi-marginal, multi-asset martingale optimal transport (MOT) problem. By extending duality and existence results to arbitrary numbers of assets and time periods, we have provided both theoretical and practical foundations for robust pricing and hedging of complex financial derivatives. Our analysis demonstrates that, under mild regularity and irreducibility conditions, dual optimizers exist and can be constructed, enabling the replication and sub-/super-hedging of path-dependent payoffs in multi-dimensional settings.

On the computational side, we have shown that the primal-dual linear programming (PDLP) approach, implemented within NVIDIA cuOpt and executed on GPU hardware, is highly effective for solving large-scale discrete MOT problems. Leveraging GPU acceleration was essential for handling the high-dimensional, path-dependent structure of these problems, enabling solution times and scalability that would be infeasible on standard CPU-based solvers. Our numerical experiments confirm that the PDLP solver achieves near-optimality and enforces intricate marginal and martingale constraints with high precision, even for challenging products such as worst-of autocallable options.

The martingale optimal transport framework thus offers a powerful and flexible tool for quantitative finance, supporting model calibration, robust pricing, and risk management in environments with limited model assumptions and rich market data. The demonstrated effectiveness of GPU-accelerated computation further expands the practical applicability of MOT to real-world, large-scale financial problems.

\paragraph{Future work} 
There remain several promising directions for further research in the study of multi-marginal, multi-asset martingale optimal transport. One avenue is to investigate how additional structural properties of the payoff function may be leveraged to refine theoretical bounds and enhance the robustness of pricing results.
Another important direction is the development of more efficient computational techniques, particularly those that exploit problem-specific features to improve scalability and reduce dimensionality in large-scale numerical implementations. Advancing these aspects could lead to significant improvements in both the theoretical understanding and practical applicability of the MOT framework in high-dimensional financial settings. We leave these and related questions for future work.

\section{Proofs of Main Results}\label{sec:proofs}

The assumption  $|c(x)| \le \sum_{t, i} v_{t,i}(x_{t,i})$ for some continuous $v_{t,i} \in L^1(\mu_{t,i})$ ensures $P(c)=D(c)$ in \eqref{duality} (see e.g. \cite{Z}). Clearly, a dual optimizer exists for $c(x)$ if and only if so does for $\tilde c (x):= c(x) - \sum_{t,i} v_{t,i}$. Thus by replacing $c$ with $\tilde c$, from now on we can (and will) assume that $c \le 0$ throughout the proofs.

\begin{proof}[Proof of Lemma \ref{L1bound*}]
Let $(\phi_n, h_n)_{n \in \N}$ be an approximating dual maximizer satisfying \eqref{dual}, \eqref{maximizing}, where  $D(c)=P(c) \in \R$ and $c \le 0$ without loss of generality. Denote $\phi_{t,n}^\oplus (x_t) =\sum_{i=1}^d  \phi_{t,i,n} (x_{t,i})$,  $h_{t,n}(\bar x_t) = \big{(} h_{t,1,n}(\bar x_t),..., h_{t,d,n}(\bar x_t)\big{)}$. Define
\begin{align}
\label{chidefinition}
\chi_{t,n}(x_{t}):= \sup_{x_1,...,x_{t-1}} \sum_{s=1}^{t-1} \big( \phi_{s,n}^\oplus (x_s) + h_{s,n}(\bar x_s) \cdot \Delta x_s \big)
\end{align}
with the convention $\chi_{1,n} = \chi_{N+1, n} \equiv 0$. Notice $\chi_{t,n}$ is a convex function on $\R^d$, since it is a supremum of affine functions of $x_t$. We claim
\begin{align}\label{chiineq}
\chi_{t,n} \le \chi_{t+1,n} - \phi_{t,n}^\oplus \ \text{ for all } t \in [N] \text{ and } n \in \N.
\end{align}
This inequality can be shown as follows:
\begin{align*}
\chi_{t+1,n}(x_{t+1}) &= \sup_{x_1,...,x_{t}} \sum_{s=1}^{t} \big( \phi_{s,n}^\oplus (x_s) + h_{s,n}(\bar x_s) \cdot \Delta x_s \big) \\
& \ge \sup_{\substack{x_1,...,x_{t-1} \\ x_{t} = x_{t+1}}} \sum_{s=1}^{t-1} \big( \phi_{s,n}^\oplus (x_s) + h_{s,n}(\bar x_s) \cdot \Delta x_s \big) +  \phi_{t,n}^\oplus (x_{t+1})\\
&= \chi_{t,n}(x_{t+1}) + \phi_{t,n}^\oplus (x_{t+1}),
\end{align*}
while the inequality $\chi_{N,n} \le -\phi^\oplus_{N,n}$ directly follows from $c \le 0$. Let $\mu^\otimes_t = \mu_{t,1} \otimes \dots \otimes \mu_{t,d}$ denote the product measure. We claim the following bound
\begin{align}\label{intbypart}
\int \chi_{t, n} \,d(\mu_t^\otimes - \mu_{t-1}^\otimes) \le C \ \text{ for all } t=2,...,N \text{ and } n \in \N,
\end{align}
where $C$ is independent of $n$ throughout the proof. 

To show \eqref{intbypart}, observe that repeated application of \eqref{chiineq} gives
\begin{align*}
\mu_t^\otimes(\chi_{t,n}) &\le \mu_t^\otimes(\chi_{t+1,n}) - \mu_t^\otimes(\phi_{t,n}^\oplus) \\
&\le \mu_{t+1}^\otimes(\chi_{t+1,n}) - \mu_t^\otimes(\phi_{t,n}^\oplus) \\
&\le \mu_{t+1}^\otimes(\chi_{t+2,n}) - \mu_{t+1}^\otimes (\phi_{t+1,n}^\oplus) - \mu_t^\otimes(\phi_{t,n}^\oplus) \\
&\le \dots \le - \sum_{s=t}^N \mu_s^\otimes (\phi_{s,n}^\oplus), 
\end{align*}
where the second inequality is due to $\mu_t^\otimes \preceq_c \mu_{t+1}^\otimes$ and convexity of $\chi$. Similarly, 
\begin{align*}
\mu_{t-1}^\otimes(\chi_{t,n}) &\ge \mu_{t-1}^\otimes(\chi_{t-1,n}) + \mu_{t-1}^\otimes(\phi_{t-1,n}^\oplus) \\
&\ge \mu_{t-2}^\otimes(\chi_{t-1,n}) + \mu_{t-1}^\otimes(\phi_{t-1,n}^\oplus) \\
&\ge \mu_{t-2}^\otimes(\chi_{t-2,n}) + \mu_{t-2}^\otimes(\phi_{t-2,n}^\oplus) + \mu_{t-1}^\otimes(\phi_{t-1,n}^\oplus) \\
&\ge \dots \ge  \sum_{s=1}^{t-1} \mu_s^\otimes (\phi_{s,n}^\oplus).
\end{align*}
The two inequalities, along with $\mu_{t-1}^\otimes(\chi_{t,n}) \le \mu_{t}^\otimes(\chi_{t,n})$ due to $\mu_{t-1}^\otimes \preceq_c \mu_{t}^\otimes$, imply
\begin{align}\label{chiineq2}
\int \chi_{t, n} \,d(\mu_t^\otimes - \mu_{t-1}^\otimes) \le  - \sum_{s=1}^{N} \mu_s^\otimes (\phi_{s,n}^\oplus) = - \mu (\phi_n)
\end{align}
which, in conjunction with \eqref{maximizing},  yields \eqref{intbypart} as claimed.

We can now obtain local uniform boundedness of $\{\chi_{t,n}\}_n$ using \eqref{intbypart} and \cite[Proposition A.1]{Lim23}. We need to meet the proposition's second condition. For this, fix any $a \in I_1 := I_{1,1} \times \dots \times I_{1,d}$, and let $L_{2,n} : \R^d \to \R$ be an affine function satisfying $L_{2,n} \le \chi_{2,n}$ and $L_{2,n}(a) = \chi_{2,n}(a)$. By linearity, we have $L_{2,n}(x_2) = \nabla L_{2,n}(x_1) \cdot (x_2 - x_1) + L_{2,n}(x_1)$. This allows us to modify the given approximating dual maximizer $(\phi_n, h_n)_{n \in \N}$ by replacing $\phi_{1,n}^\oplus(x_1)$ with $\phi_{1,n}^\oplus(x_1)- L_{2,n}(x_1)$, $\phi_{2,n}^\oplus(x_2)$ with $\phi_{2,n}^\oplus(x_2) + L_{2,n}(x_2)$, and $h_{1,n}(x_1)$ with $h_{1,n}(x_1) - \nabla L_{2,n}(x_1)$ ($\nabla L_{2,n}$ is constant and does not depend on $x_1$). Notice this yields $\chi_{2,n} \ge 0$  and $\chi_{2,n}(a) =0$. We can continue subtracting appropriate linear functions $L_{t,n}$, $t=2,...,N$, and achieve
\begin{align}\label{convex0}
\chi_{t,n} \ge 0 \ \text{ and } \ \chi_{t,n}(a) =0 \quad \text{for all } n \in \N \text{ and } t \in [N].
\end{align}
Note that the modifications have no effect on the value $\mu(\phi_n)$.

Next, let $\{\epsilon_k\}_k$ be a positive decreasing sequence tending to zero as $k \to \infty$, and write $I_{t,i} = ]a_{t,i}, b_{t,i}[$ where $-\infty \le a_{t,i} < b_{t,i} \le +\infty$. Then we define the compact interval $J_{t,i,k}:=[c_{t,i,k}, d_{t,i,k}]$ for $t \in [N-1]$, $i \in [d]$ and $k \in \N$ as follows:
\begin{align}\label{J}
&\text{If } a_{t,i} > -\infty, \text{ then define $c_{t,i,k }$ by the following rule:}\\
&\ \mu_{t+1,i} (a_{t,i})=0 \Rightarrow c_{t,i,k }:= a_{t,i} + \epsilon_k; \ \ \mu_{t+1,i} (a_{t,i}) > 0 \Rightarrow c_{t,i,k} := a_{t,i}; \nn\\
&\text{If } b_{t,i} < +\infty, \text{ then define $d_{t,i,k }$ by the following rule:}\nn\\
&\ \mu_{t+1,i} (b_{t,i})=0 \Rightarrow d_{t,i,k }:= b_{t,i} -  \epsilon_k; \ \ \mu_{t+1,i} (b_{t,i})>0 \Rightarrow d_{t,i,k }:= b_{t,i};  \nn\\
&\text{If }  a_{t,i} = -\infty, \text{ then } c_{t,i,k} := -1/ \epsilon_k; \nn\\
&\text{If }  b_{t,i} = +\infty, \text{ then }  d_{t,i,k} := +1/\epsilon_k. \nn
\end{align}
Set $J_{0,i,k} := J_{1,i,k} $. For example, if $\mu_{t+1,i} (a_{t,i}) = 0$ and $\mu_{t+1,i} (b_{t,i}) > 0$, then $J_{t,i,k} = [a_{t,i} + \epsilon_k, b_{t,i}]$. Let $\epsilon_1$ be so small so that $\mu_{t,i} (J_{t,i,1})>0$, $\mu_{t+1,i} (J_{t,i,1}) > 0$ for every $t,i$. Observe that $J_{t,i,k} \nearrow  J_{t,i}$ as $k \to \infty$. Let $J_{t,k} := J_{t,1,k} \times J_{t,2,k} \times ... \times J_{t,d,k}$. Then \cite[Proposition A.1]{Lim23} with \eqref{intbypart} and \eqref{convex0} yields that there exists $M_{k} \ge 0$ for each $k \in \N$ such that 
\begin{align}\label{boundconvex}
0 \le \sup_{n \in \N} \sup_{x_t \in J_{t-1,k}}  \chi_{t,n}(x_t) \le M_{k}.
\end{align}
Next, another repeated application of \eqref{chiineq} gives
\begin{align*}
C &\ge  -\sum_{s=1}^{N} \mu_s^\otimes (\phi_{s,n}^\oplus)\\
&\ge  \mu_N^\otimes (\chi_{N,n}) -\sum_{s=1}^{N-1} \mu_s^\otimes (\phi_{s,n}^\oplus) \\
&\ge  \mu_{N-1}^\otimes (\chi_{N,n}) -\sum_{s=1}^{N-1} \mu_s^\otimes (\phi_{s,n}^\oplus) \\
&=\,  \mu_{N-1}^\otimes (\chi_{N,n}) -  \mu_{N-1}^\otimes (\chi_{N-1,n}) +  \mu_{N-1}^\otimes (\chi_{N-1,n}) -\sum_{s=1}^{N-1} \mu_s^\otimes (\phi_{s,n}^\oplus) \\
&=\,  \parallel \chi_{N,n} -\chi_{N-1,n} - \phi_{N-1,n}^\oplus \parallel_{L^1( \mu_{N-1}^\otimes )}  +  \mu_{N-1}^\otimes (\chi_{N-1,n}) -\sum_{s=1}^{N-2} \mu_s^\otimes (\phi_{s,n}^\oplus) \\
&\ge\,  \parallel \chi_{N,n} -\chi_{N-1,n} - \phi_{N-1,n}^\oplus \parallel_{L^1( \mu_{N-1}^\otimes )}  +  \mu_{N-2}^\otimes (\chi_{N-1,n}) -\sum_{s=1}^{N-2} \mu_s^\otimes (\phi_{s,n}^\oplus) \\
&\ge \dots \ge  \sum_{s=2}^N \parallel \chi_{s,n} -\chi_{s-1,n} - \phi_{s-1,n}^\oplus \parallel_{L^1( \mu_{s-1}^\otimes )}
\end{align*}
where the third and sixth inequality is due to the  convexity of $\chi$ with $\mu_t^\otimes \preceq_c \mu_{t+1}^\otimes$, and the fifth equality is by the nonnegativity $\chi_{t,n} -\chi_{t-1,n} - \phi_{t-1,n}^\oplus \ge 0$ from \eqref{chiineq}. 

On the other hand, the nonpositivity $c \le 0$ in \eqref{dual} yields
\begin{align*}
\sum_{s=1}^{N} \big( \phi_{s,n}^\oplus (x_s) + h_{s,n}(\bar x_s) \cdot \Delta x_s \big) \le \chi_{N,n} (x_N) + \phi^\oplus_{N,n} (x_N) \le 0
\end{align*}
where the first inequality follows by taking supremum over $x_1,...,x_{N-1}$ (recall \eqref{chidefinition} and the convention $h_{N,n} \equiv 0$). Integrating with any $\pi \in {\rm VMT}(\mu)$ yields $\parallel \chi_{N,n} (x_N) + \phi^\oplus_{N,n}\parallel_{L^1( \mu_{N}^\otimes )} \le -\sum_{s=1}^{N} \mu_s^\otimes (\phi_{s,n}^\oplus) \le C$. We therefore deduce
\be\label{L1bound}
\parallel \chi_{t+1,n} -\chi_{t,n} -
 \phi_{t,n}^\oplus \parallel_{L^1( \mu_{t}^\otimes )}\, \le C \ \text{ for all } n \in \N, t \in [N].
\ee
Now for each $k \in \N$, let $\mu_{t,i,k}$ be the restriction of $\mu_{t,i}$ on $J_{t-1,i,k}$ (where $J_{0,i,k} := J_{1,i,k}$) then normalized to be a probability distribution, as defined in Lemma \ref{L1bound*}. Let $\mu^\otimes_{t,k} = \otimes_i \mu_{t,i,k}$, so that $\mu^\otimes_{t,k}(J_{t-1,k})=1$ where $J_{t-1,k} := \otimes_i J_{t-1,i,k}$. Define
\begin{align} %
    v_{t,i,k,n} := \int \phi_{t,i,n} \,d\mu_{t,i,k}, \quad t \in [N], i \in [d], k\in \N, n \in \N. \nn
\end{align}
To complete the proof of Lemma \ref{L1bound*}, that is to show for each $k \in \N$
\begin{align}\label{supbound}
\sup_n \parallel \phi_{t,i,n} - v_{t,i,k,n} \parallel_{L^1(\mu_{t,i,k})} \, \le C,
\end{align}
observe that \eqref{boundconvex}, \eqref{L1bound} and the fact $J_{t-1,k} \subset J_{t,k}$ together imply
\begin{align*}
&\parallel \phi_{t,n}^\oplus \parallel_{L^1(\mu^\otimes_{t,k})} \\ 
&\le\,  \parallel \chi_{t+1,n} -\chi_{t,n} -
 \phi_{t,n}^\oplus \parallel_{L^1(\mu^\otimes_{t,k})} + \parallel \chi_{t,n} - \chi_{t+1,n}  \parallel_{L^1(\mu^\otimes_{t,k})} \\
&\le C + M_k =: C
\end{align*}
where we used $\chi_{t+1,n} -\chi_{t,n} -
 \phi_{t,n}^\oplus \ge 0$ to get the bound $ \parallel \chi_{t+1,n} -\chi_{t,n} -
 \phi_{t,n}^\oplus \parallel_{L^1(\mu^\otimes_{t,k})}\, \le C$ from \eqref{L1bound}. From this, we obtain the bound
\begin{align}\label{boundv}
\bigg| \sum_{i=1}^d v_{t,i,k,n} \bigg| \le \, \parallel \phi_{t,n}^\oplus \parallel_{L^1(\mu^\otimes_{t,k})}  \, \le C \, \text{ for all } n,
\end{align}
where the first inequality is by Jensen's inequality.
Next, because $\phi_{t,n}^\oplus \le M_{k}$ on $J_{t,k}$ by \eqref{chiineq}, \eqref{convex0} and \eqref{boundconvex}, by taking supremum over $J_{t,k}$, we have
\begin{align*}
\sum_{i=1}^d \sup_{x_{t,i} \in J_{t,i,k}} \phi_{t,i,n} (x_{t,i}) \le M_{k}  \, \text{ for all } n.
\end{align*}
In particular, since $v_{t,i,k,n} \le \sup_{x_{t,i} \in J_{t,i,k}} \phi_{t,i,n} (x_{t,i})$, 
\begin{align}\label{upperbd}
\sup_{x_{t,1} \in J_{t,1,k}} \phi_{t,1,n} (x_{t,1})+ \sum_{i=2}^d v_{t,i,k,n}  \le M_k.
\end{align}
Define  $\hat v_{t,1,k,n} := -\sum_{i=2}^d v_{t,i,k,n}$. Since $\phi_{t,n}^\oplus  \le M_{k}$ on $J_{t,k}$, we have
\begin{align*}
C \ge \,\,\parallel M_{k} - \phi_{t,n}^\oplus \parallel_{L^1(\mu^\otimes_{t,k})}\, = M_k - \int ( \phi_{t,1,n} + \sum_{i=2}^d v_{t,i,k,n}) d\mu_{t,1,k}.
\end{align*}
With \eqref{upperbd}, this implies that $ \sup_n \parallel  \phi_{t,1,n} - \hat v_{t,1,k,n} \parallel_{L^1(\mu_{t,1,k})}$ is bounded, and then by \eqref{boundv}, $ \sup_n \parallel  \phi_{t,1,n} - v_{t,1,k,n} \parallel_{L^1(\mu_{t,1,k})}$ is bounded. This proves \eqref{supbound}. 
\end{proof}

We turn to the proof of Proposition \ref{ptwiseconverge}.
\begin{proof}[Proof of Proposition \ref{ptwiseconverge} ]
Lemma \ref{L1bound*} and Koml{\'o}s lemma (which states that every $L^1$-bounded sequence of real functions contains a subsequence such that the arithmetic means of all its subsequences converge almost everywhere) implies that for each $k \in \N$, there exists a subsequence $\{\phi_{t,i,k,n}\}_n$ of $\{\phi_{t,i,n}\}_n$ such that

(i) $\{\phi_{t,i,k+1,n}\}_n$ is a further subsequence of $\{\phi_{t,i,k,n}\}_n$, and

(ii) $\tilde \phi_{t,i,k,n}(x_{t,i}) - \tilde v_{t,i,k,n}$ converges $\mu_{t,i,k}$ - a.s. \, as \, $n \to \infty$,\\
where $v_{t,i,k,n} = \int \phi_{t,i,k,n} \,d\mu_{t,i,k}$, $\hat v_{t,1,k,n} = -\sum_{i=2}^d v_{t,i,k,n}$, $\tilde v_{t,1,k,n} =  \frac{1}{n} \sum_{m=1}^n \hat v_{t,1,k,m}$, $\tilde v_{t,i,k,n} =  \frac{1}{n} \sum_{m=1}^n v_{t,i,k,m}$ for $i \ge 2$, and $\tilde \phi_{t,i,k,n}= \frac{1}{n} \sum_{m=1}^n \phi_{t,i,k,m}$. Note that for each $k$, our choice of a subsequence index can be made identical for every $t, i$,  since there are finitely many indices of $t,i$. Then we select the diagonal sequence  
\[
\Phi_{t,i,n} := \phi_{t,i,n,n}
\]
 and again define 
  \begin{align*}
 w_{t,i,k,n} &= \int \Phi_{t,i,n} d\mu_{t,i,k}, \ \ \tilde w_{t,i,k,n} =  \frac{1}{n} \sum_{m=1}^n w_{t,i,k,m},  \q 2\le i \le d,\\
 \hat w_{t,1,k,n} &= -\sum_{i=2}^d w_{t,i,k,n}, \ \  \tilde w_{t,1,k,n} =  \frac{1}{n} \sum_{m=1}^n \hat w_{t,1,k,m}, \ \ \tilde \Phi_{t,i,n}(x_{t,i}) = \frac{1}{n} \sum_{m=1}^n \Phi_{t,i,m} (x_{t,i}).
  \end{align*}
 We then claim:
 \begin{align}\label{convergePhi}
\tilde \Phi_{t,i,n}(x_{t,i})  - \tilde w_{t,i,1,n} \, \text{converges } \, \mu_{t,i} - a.s. \, \text{ for all } t \in [N], i \in [d].
\end{align}
Note that the dependence on $k$ has now been removed. To prove the claim, since $\{\Phi_{t,i,n}\}_n$ is a subsequence of $\{\phi_{t,i,k,n}\}_n$ for every $k\in \N$,   Koml{\'o}s lemma implies
 \begin{align}\label{convergePhik}
\tilde \Phi_{t,i,n}(x_{t,i})  - \tilde w_{t,i,k,n} \ \text{converges } \, \mu_{t,i,k} - a.s. \, \text{ for all } t \text{ and } i. 
\end{align}
In particular, both $\{\tilde \Phi_{t,i,n}(x_{t,i})  - \tilde w_{t,i,1,n}\}_n$ and $\{\tilde \Phi_{t,i,n}(x_{t,i})  - \tilde w_{t,i,k,n}\}_n$ converge $ \mu_{t,i,1}$-a.s. as $n\to \infty$, so does the difference $ \{\tilde w_{t,i,1,n} -  \tilde w_{t,i,k,n}\}_n$ for each  $k$. Hence \eqref{convergePhi} follows from \eqref{convergePhik}. The identity $\sum_{i=1}^d \tilde w_{t,i,1,n} = 0$ then allows to replace the approximating dual maximizer $(\phi_n, h_n)_{n}$ by $(\psi_n, \tilde h_n)_{n}$, where $\psi_{t,i,n} :=  \tilde \Phi_{t,i,n}(x_{t,i})  - \tilde w_{t,i,1,n}$ and $\tilde h_{t,i,n}$ is the corresponding Ces{\`a}ro mean of the subsequence of $(h_{t,i,n})_n$ chosen consistently with the selection of $\{\Phi_{t,i,n}\}_n$ out of $\{ \phi_{t,i,n}\}_n$.
\end{proof}

Finally, we turn to the proof of Theorem \ref{main}.

\begin{proof}[Proof of Theorem \ref{main}]

By Proposition \ref{ptwiseconverge}, we have an approximating dual maximizer $(\phi_n, h_n)_{n}$ such that
\be\label{asconverge}
\{ \phi_{t,i,n} \}_n \text{ converges to a function } \phi_{t,i} \ \mu_{t,i}-a.s. \text{ as } n \to \infty \ \text{ for each $t \in [N]$ and $i \in [d]$}.
\ee
We shall prove the convergence of the convex functions $\{\chi_{t,n}\}_n$ defined in \eqref{chidefinition}.

In the proof of Lemma \ref{L1bound*}, we obtained an approximating dual maximizer (say $(\phi^0_n, h^0_n)_{n}$) such that the associated $\{\chi^0_{t,n}\}_n$ satisfied the local bound \eqref{boundconvex}. As observed in the proof of Proposition \ref{ptwiseconverge}, the approximating dual maximizer $(\phi_n, h_n)_{n}$ satisfying \eqref{asconverge} can be taken as a Ces{\`a}ro mean of a suitable subsequence of $(\phi^0_n, h^0_n)_{n}$. This implies the upper bound in \eqref{boundconvex} continues to hold
\begin{align}\label{boundconvex1}
\sup_n \chi_{t,n} \le M_{k} \ \ \text{on} \ \ J_{t-1,k}.
\end{align}
We shall modify $(\phi_n, h_n)_{n}$ so that \eqref{convex0} holds while \eqref{boundconvex1} is retained. Let $I_t := I_{t,1} \times \dots \times I_{t,d}$. Fix any $a \in I_1$ and define $\phi^\oplus_t := \sum_{i} \phi_{t,i}$. First, \eqref{asconverge}  implies that there exists $a_t \in I_t$ for every $t \in [N]$ such that 
\be
\lim_{n \to \infty} \phi_{t,n}^\oplus (a_t) = \phi_t^\oplus (a_t) \in \R.
\ee
In view of \eqref{chiineq} which gives $\phi^\oplus_{1,n} \le \chi_{2,n}$, this implies
\be\label{goodlowbound}
\inf_n \chi_{2,n} (a_1) > -\infty.
\ee
On the other hand, since $I_1 = {\rm int} (J_1)$ and $J_{1,k} \nearrow J_{1}$ (see \eqref{J} for the definition of compact intervals $J_{t,i,k}$), for large enough $k$ we have $\{a, a_1\} \subset {\rm int}(J_{1,k})$. Now \eqref{boundconvex1}, \eqref{goodlowbound} imply that both $\chi_{2,n}(a)$ and $\nabla \chi_{2,n} (a)$ are uniformly bounded in $n$, where $\nabla \chi_{2,n} (a) \in \partial \chi_{2,n} (a)$ is  a subgradient  of the convex function $\chi_{2,n}$ at $a$. Hence by taking a subsequence, we can assume that $\{\chi_{2,n}(a)\}_n$ and $\{\nabla \chi_{2,n} (a)\}_n$ both converge. As in the proof of Lemma \ref{L1bound*}, define an affine function $L_{2,n}(y) = \chi_{2,n}(a) + \nabla \chi_{2,n} (a) \cdot (y-a)$, and replace $\phi_{1,n}^\oplus(x_1)$ with $\phi_{1,n}^\oplus(x_1) - L_{2,n}(x_1)$, $\phi_{2,n}^\oplus(x_2)$ with $\phi_{2,n}^\oplus(x_2) + L_{2,n}(x_2)$, and finally $h_{1,n}(x_1)$ with $h_{1,n}(x_1) - \nabla \chi_{2,n}(a)$. This yields $\chi_{2,n}(a) = \nabla \chi_{2,n}(a) = 0$ for all $n$, while the a.s. convergence of $\phi_{t,i, n}$ and the bound \eqref{boundconvex1} are retained. Next, the inequality $\chi_{3,n} \ge \chi_{2,n} + \phi^\oplus_{2,n}$ with $\chi_{2,n} \ge 0$ yields $\inf_n \chi_{3,n}(a_2) > -\infty$. Thus we can repeat the argument and achieve the normalization \eqref{convex0}, while the a.s. convergence of $\phi_{t,i, n}$ and the bound \eqref{boundconvex1} are retained. Thanks to the local bound \eqref{boundconvex} and \eqref{intbypart}, we can now apply \cite[Proposition A.1]{Lim23} to deduce the pointwise convergence $\chi_{t,n} \to \chi_t$ on $J_{t-1}$.

As the limit function $(\phi_{t,i})_{t,i}$ is only well defined $\mu_{t,i}$-a.s., define $\phi_{t,i} := -\infty$ on a $\mu_{t,i}$-null set (which includes $\R \setminus I_{t,i}$), so that they are defined everywhere on $\R$. We now show that there exist functions $h_t = (h_{t,i})_{i} : \R^{td} \to \R^d$ for all $t \in [N]$ with $h_N \equiv 0$, such that 
\begin{align}\label{duallimit}
\sum_{t=1}^{N} \big( \phi_{t}^\oplus (x_t) + h_{t}(\bar x_t) \cdot \Delta x_t \big) \le c(x) \ \ \text{for all } x \in \R^{Nd}.
\end{align}
For any function $f : \R^d \to \R \cup \{+\infty\}$ which is bounded below by an affine function, let ${\rm conv}[f]:\R^d \to \R \cup \{+\infty\}$ denote the lower semi-continuous convex envelope of $f$, that is the supremum of all affine functions $l$ satisfying $ l \le f$ (If there is no such $l$, let ${\rm conv}[f] \equiv -\infty$.) We will inductively obtain $h_{N-1}, h_{N-2},...,h_1$. Let us rewrite \eqref{dual} as
\begin{align}\label{a1}
\sum_{t=1}^{N-1} \big( \phi_{t,n}^\oplus (x_t) + h_{t,n}(\bar x_t) \cdot \Delta x_t \big) \le c(x) -  \phi_{N,n}^\oplus (x_N).
\end{align}
Define $H_{N-1,n}(\bar x_{N-1},x_N) = {\rm conv}[c(\bar x_{N-1},\,\cdot\,) - \phi_{N,n}^\oplus (\,\cdot\,)](x_N)$. We then  have
\begin{align*}
\sum_{t=1}^{N-1} \big( \phi_{t,n}^\oplus (x_t) + h_{t,n}(\bar x_t) \cdot \Delta x_t \big) \le H_{N-1,n}(\bar x_{N-1},x_N) \le c(x) -  \phi_{N,n}^\oplus (x_N)
\end{align*}
because the left hand side is affine in $x_N$. If we let $x_N = x_{N-1}$, we get
\begin{align}\label{a2}
\sum_{t=1}^{N-2} \big( \phi_{t,n}^\oplus (x_t) + h_{t,n}(\bar x_t) \cdot \Delta x_t \big) \le H_{N-1,n}(\bar x_{N-1},x_{N-1}) -  \phi_{N-1,n}^\oplus (x_{N-1}).
\end{align}
We see that \eqref{a1} and \eqref{a2} exhibit a similar structure. With the convention $H_{N,n} (x) := c(x)$,
 this allows us to inductively deduce, backward in $t$,
 \begin{align}\label{a3}
\sum_{s=1}^{t} \big( \phi_{s,n}^\oplus (x_s) + h_{s,n}(\bar x_s) \cdot \Delta x_s \big) \le H_{t+1,n}(\bar x_{t+1},x_{t+1}) -  \phi_{t+1,n}^\oplus (x_{t+1})
\end{align}
for $t=1,...,N-1$, where
\begin{align}\label{a4}
H_{t,n}(\bar x_{t},x_{t+1}) := {\rm conv}[H_{t+1,n}(\bar x_{t},\,\cdot\,) - \phi_{t+1,n}^\oplus (\,\cdot\,)](x_{t+1})
\end{align}
with an abuse of notation $H_{t+1,n}(\bar x_{t},x_{t+1}) := H_{t+1,n}(\bar x_{t+1},x_{t+1}) $. 

Now by dropping the index $n$, we analogously define
\begin{align}\label{a5}
H_{t}(\bar x_{t},x_{t+1}) := {\rm conv}[H_{t+1}(\bar x_{t},\,\cdot\,) - \phi_{t+1}^\oplus (\,\cdot\,)](x_{t+1})
\end{align}
with $H_{N} (x) := c(x)$. Next, since the $\limsup$ of convex functions is convex, in conjunction with the almost sure convergence \eqref{asconverge}, we have
\begin{align*}
\limsup_{n \to \infty} H_{N-1,n}(\bar x_{N-1},x_N) 
&= \limsup_{n \to \infty} {\rm conv}[c(\bar x_{N-1},\,\cdot\,) - \phi_{N,n}^\oplus (\,\cdot\,)](x_N) \\
&\le {\rm conv}[\limsup_{n \to \infty}\big{(} c(\bar x_{N-1},\,\cdot\,) - \phi_{N,n}^\oplus (\,\cdot\,) \big{)}](x_N)  \\
&\le {\rm conv}[c(\bar x_{N-1},\,\cdot\,) - \phi_{N}^\oplus (\,\cdot\,)](x_N) \\
&= H_{N-1}(\bar x_{N-1},x_N).
\end{align*}
This allows us to inductively deduce, for $t=1,...,N-1$,
\begin{align*}
\limsup_{n \to \infty} H_{t,n}(\bar x_{t},x_{t+1}) 
&= \limsup_{n \to \infty} {\rm conv}[H_{t+1,n}(\bar x_{t},\,\cdot\,) - \phi_{t+1,n}^\oplus (\,\cdot\,)](x_{t+1}) \\
&\le {\rm conv}[\limsup_{n \to \infty}\big{(} H_{t+1,n}(\bar x_{t},\,\cdot\,) - \phi_{t+1,n}^\oplus (\,\cdot\,)\big{)}](x_{t+1})  \\
&\le {\rm conv}[H_{t+1}(\bar x_{t},\,\cdot\,) - \phi_{t+1}^\oplus (\,\cdot\,)](x_{t+1}) \\
&= H_{t}(\bar x_{t},x_{t+1}).
\end{align*}
Let us discuss continuity of the convex function $x_{t+1} \mapsto H_{t}(\bar x_{t}, x_{t+1})$. The following inequality from \eqref{a5}
\be
H_{t}(\bar x_{t},x_{t+1}) \le H_{t+1}(\bar x_{t+1}, x_{t+1}) - \phi_{t+1}^\oplus (x_{t+1}) \nn 
\ee 
becomes $H_{N-1}(\bar x_{N-1},x_{N}) \le c(x) - \phi_{N}^\oplus (x_{N})$ when $t=N-1$. Then the $\mu_N^\otimes$-a.s. finiteness of $\phi_{N}^\oplus$ implies, by convexity, $H_{N-1}(\bar x_{N-1},x_{N}) < \infty$ if $x_N \in J_{N-1}$.  Backward induction in $t$ then gives $H_{t}(\bar x_{t},x_{t+1}) < \infty$ if $x_{t+1} \in J_{t}$. This implies that if there exists $y_0 \in \R^d$ such that $H_{t}(\bar x_{t}, y_0) > -\infty$, then $y \mapsto H_{t}(\bar x_{t}, y)$ is real-valued and continuous in $J_{t}$. This being observed, now \eqref{a3},  \eqref{a4} gives
\begin{align*}
\phi_{1,n}^\oplus (x_1) + h_{1,n}(x_1) \cdot \Delta x_1 
 \le H_{1,n}( x_{1},x_{2}) \le H_{2,n}( x_{1},x_{2}, x_{2}) -  \phi_{2,n}^\oplus (x_{2}).
\end{align*}
Letting $x_2=x_1$ gives $\phi_{1,n}^\oplus (x_1) \le H_{1,n}( x_{1},x_{1})$. Taking $\limsup$ yields
\begin{align*}
\phi_{1}^\oplus (x_1) \le H_{1}( x_{1},x_{1}) \ \text{ and } \ H_{1}( x_{1},x_{2}) \le H_{2}( x_{1},x_{2},x_{2}) -  \phi_{2}^\oplus (x_{2}). 
\end{align*}
Define $A_t:= \{ x_t \in \R^d \, | \, \phi_{t}^\oplus (x_t) \in \R \}$ for each $t \in [N]$, and note that $A_t \subset I_t$. Since $x_2 \mapsto H_1(x_1,x_2)$ is continuous in $J_1$ for every $x_1 \in A_1$, the subdifferential $\partial H_1(x_1, \,\cdot\,) (x_2)$ is nonempty, convex and compact for every $x_2 \in I_1 = {\rm int}(J_1)$. This allows us to choose a measurable function $h_1 : A_1 \to \R^d$ satisfying $h_1(x_1) \in \partial H_1(x_1, \,\cdot\,)(x_1)$. Then for $x_1 \in A_1$, we have
\begin{align*}
\phi_{1}^\oplus (x_1) + h_1(x_1) \cdot (x_2-x_1) &\le H_1(x_1,x_1) + h_1(x_1) \cdot (x_2-x_1) \\
 &\le H_1(x_1,x_2) \\
 &\le H_{2}( x_{1},x_{2}, x_{2}) -  \phi_{2}^\oplus (x_{2}). 
 \end{align*}
In particular, for $x_1 \in A_1$ and $x_2 \in A_2$, it holds $H_{2}( x_{1},x_{2}, x_{2}) > -\infty$. Hence again we can choose $h_2 : A_1 \times A_2 \to \R^d$ that satisfies $h_2(x_1,x_2) \in \partial H_2 (x_1, x_2, \,\cdot\,)(x_2)$. Then for every $x_1 \in A_1$ and $x_2 \in A_2$, we have
\begin{align*}
&\phi_{1}^\oplus (x_1) +  \phi_{2}^\oplus (x_{2}) + h_1(x_1) \cdot (x_2-x_1) + h_2(x_1,x_2) \cdot (x_3-x_2) \\&\le H_{2}( x_{1},x_{2}, x_{2}) + h_2(x_1,x_2) \cdot (x_3-x_2)\\
&\le H_{2}( x_{1},x_{2}, x_{3})\\
&\le H_{3}( x_{1},x_{2}, x_{3}, x_{3}) - \phi_{3}^\oplus (x_{3}).
 \end{align*}
By induction in $t$, we obtain $h_t : A_1 \times \dots \times A_t \to \R^d$ (with $h_N \equiv 0$), satisfying \eqref{duallimit} as desired. We may define $h_t = 0$ in $\R^{td} \setminus A_1 \times \dots \times A_t$, noting that the left hand side of \eqref{duallimit} is $-\infty$ if $x_t \notin A_t$ for some $t$.

The last step is to show the following claim: for any functions $h_t : \R^{td} \to \R^d$, $t=1,...,N$ with $h_N \equiv 0$ satisfying  \eqref{duallimit} (whose existence was just shown), and for any minimizer $\pi^*$ for the VMOT problem \eqref{VMOT}, it holds
\begin{align}\label{pointwisedualeq}
\sum_{t=1}^{N} \big( \phi_{t}^\oplus (x_t) + h_{t}(\bar x_t) \cdot \Delta x_t \big) = c(x), \quad \pi^* - a.s..
\end{align}
In other words, every minimizer $\pi^*$ is concentrated on the contact set
\[
\Gamma := \bigg\{x \in \R^{Nd} \,\bigg|\, \sum_{t=1}^{N} \big( \phi_{t}^\oplus (x_t) + h_{t}(\bar x_t) \cdot \Delta x_t \big) = c(x) \bigg\}
\]
whenever $\{h_t\}_t$ satisfies \eqref{duallimit}. This will complete the proof of the theorem.

The proof of this claim was provided in \cite{bnt} for the single-asset, two-period case $(d, N) = (1, 2)$, and we build upon that approach to extend to the general case. To begin, recall $\phi_{t,i,n} \to \phi_{t,i}$ $\mu_{t,i}$-a.s. and $\chi_{t,n} \to \chi_t$ in $J_{t-1}$ where $\chi_{t,n}$ is defined in \eqref{chidefinition} with $\chi_t$ being its limit. For any $\pi \in {\rm VMT}(\mu)$ (not necessarily an optimizer), we have $c \in L^1(\pi)$ by the assumption of Theorem \ref{main}. We claim:
\begin{align}\label{claim1}
\limsup_{n \to \infty}& \int \sum_{t=1}^{N} \big( \phi_{t,n}^\oplus (x_t) + h_{t,n}(\bar x_t) \cdot \Delta x_t \big)  d\pi  \\
&\le \int \sum_{t=1}^{N} \big( \phi_{t}^\oplus (x_t) + h_{t}(\bar x_t) \cdot \Delta x_t \big) d\pi. \nn
\end{align}
To see how the claim implies \eqref{pointwisedualeq}, let $\pi^*$ be any minimizer for \eqref{VMOT} (which exists by the assumption on $c$). Then %
$D(c) = P(c) = \int c \,d\pi^*$, hence 
\begin{align*}
P(c)&=\lim_{n \to \infty}  \int \sum_{t=1}^{N} \big( \phi_{t,n}^\oplus (x_t) + h_{t,n}(\bar x_t) \cdot \Delta x_t \big) d\pi^* \\ 
&\le \int \sum_{t=1}^{N} \big( \phi_{t}^\oplus (x_t) + h_{t}(\bar x_t) \cdot \Delta x_t \big) d\pi^* \\ 
&\le  \int c(x) \,d\pi^* = P(c)
\end{align*}
by Fatou's lemma, yielding equality throughout. Notice that this implies \eqref{pointwisedualeq}.

To prove \eqref{claim1}, fix any $\pi \in {\rm VMT}(\mu)$. Denoting $\pi = {\rm Law}(X)$ where $X=(X_1,...,X_N)$ is an $\R^d$-valued martingale under $\pi$, let $\pi_t := {\rm Law}(X_t)$. By following the argument which establishes \eqref{L1bound}, but this time utilizing the convex order $\pi_t \preceq_c \pi_{t+1}$ instead of $\mu^\otimes_t \preceq_c \mu^\otimes_{t+1}$, by \eqref{maximizing} and \eqref{chiineq}, we can deduce
\be\label{L1bound2}
\parallel \chi_{t+1,n} -\chi_{t,n} -
 \phi_{t,n}^\oplus \parallel_{L^1( \pi_{t})} \, \le C \ \text{ for all } n \in \N, t \in [N].
\ee
Using $\phi^\oplus_{t,n} \to \phi^\oplus_{t}$, $\chi_{t,n} \to \chi_t$ and Fatou's lemma, from \eqref{L1bound2} we deduce
\begin{align*}
&\chi_{t+1} -\chi_{t} -
 \phi_{t}^\oplus \in L^1(\pi_t), \ \text{ and}\\
\limsup_{n \to \infty} \int (&\phi_{t,n}^\oplus + \chi_{t,n} - \chi_{t+1,n}) \,d\pi_t \le  \int (\phi_{t}^\oplus + \chi_{t} - \chi_{t+1} )\,d\pi_t,
\end{align*}
recalling $\chi_{1,n} = \chi_{N+1,n} \equiv 0$ and $\pi_{t} (J_{t-1}) = 1$. This allows us to deduce
\begin{align}
\limsup_{n \to \infty}  &\int \sum_{t=1}^{N} \big( \phi_{t,n}^\oplus (x_t) + h_{t,n}(\bar x_t) \cdot \Delta x_t \big) d\pi \nn \\
= \limsup_{n \to \infty}  &\int \sum_{t=1}^{N} \big( \phi_{t,n}^\oplus (x_t) + \chi_{t,n}(x_t) - \chi_{t+1,n}(x_t) \nn \\
&\ \ \q\q\q - \chi_{t,n}(x_t) + \chi_{t+1,n}(x_t)+ h_{t,n}(\bar x_t) \cdot \Delta x_t \big) d\pi \nn \\
\le  \sum_{t=1}^N \int &\big( \phi_{t}^\oplus + \chi_{t} - \chi_{t+1} \big) d\pi_t  \label{a8} \\
 + \limsup_{n \to \infty} &\int \sum_{t=1}^{N-1} \big{(}  \chi_{t+1,n}(x_t)  - \chi_{t+1,n}(x_{t+1})  + h_{t,n}(\bar x_t) \cdot \Delta x_t  \big{)} \, d\pi, \nn
\end{align}
since $ \sum_{t=1}^N \chi_{t,n}(x_{t}) = \sum_{t=1}^N \chi_{t+1,n}(x_{t+1})$. Now denote $\bar X_t = (X_1,...,X_t)$, $\pi^t := {\rm Law}(\bar X_t) \in \cP(\R^{td})$. Then we can write $\pi^{t+1} = \pi_{\bar x_t} \otimes \pi^t$, where $\pi_{\bar x_t} \in \cP(\R^d)$ is the conditional distribution of $X_{t+1}$ given $\bar X_t = \bar x_t$ under the martingale law $\pi$. Martingale property means that $\int y \, d \pi_{\bar x_t}(y) = x_t$. For each $t$, choose a sequence of functions $\xi_{t,n} : I_t \to \R^d$ satisfying $\xi_{t,n}(x_t) \in \partial \chi_{t+1,n}(x_t)$. Then we compute
\begin{align*}
&\int \sum_{t=1}^{N-1} \big{(}  \chi_{t+1,n}(x_t)  - \chi_{t+1,n}(x_{t+1})  + h_{t,n}(\bar x_t) \cdot \Delta x_t  \big{)} \, d\pi \\
&= \int \bigg( \int \big( \chi_{N,n}(x_{N-1})  - \chi_{N,n}(x_{N})  + \xi_{N-1,n}( x_{N-1}) \cdot \Delta x_{N-1}\big)  d\pi_{\bar x_{N-1}}(x_N)\\
&\q\q\q + \sum_{t=1}^{N-2} \big(\chi_{t+1,n}(x_t)  - \chi_{t+1,n}(x_{t+1})  + h_{t,n}(\bar x_t) \cdot \Delta x_t  \big{)} \bigg)  d\pi^{N-1}(\bar x_{N-1}),
\end{align*}
where we replaced $h_{N-1,n}(\bar x_{N-1})$ by $\xi_{N-1,n}(\bar x_{N-1})$ due to the martingale property
\begin{align}\label{a7}
\int h_{N-1,n}(\bar x_{N-1}) \cdot \Delta x_{N-1} \, d\pi_{\bar x_{N-1}}(x_N) = \int \xi_{N-1,n}( x_{N-1}) \cdot \Delta x_{N-1} \, d\pi_{\bar x_{N-1}}(x_N) =0. \nn
\end{align}
By definition of $\xi$, we have $\chi_{N,n}(x_{N-1})  - \chi_{N,n}(x_{N})  + \xi_{N-1,n}( x_{N-1}) \cdot \Delta x_{N-1} \le 0$. 
This allows us to disintegrate $\pi^{N-1} = \pi_{\bar x_{N-2}} \otimes \pi^{N-2}$ and repeat the same argument
\begin{align*}
&\int \sum_{t=1}^{N-1} \big{(}  \chi_{t+1,n}(x_t)  - \chi_{t+1,n}(x_{t+1})  + h_{t,n}(\bar x_t) \cdot \Delta x_t  \big{)} \, d\pi \\
&= \int\dots \int \big( \chi_{N,n}(x_{N-1})  - \chi_{N,n}(x_{N})  + \xi_{N-1,n}( x_{N-1}) \cdot \Delta x_{N-1}\big)  d\pi_{\bar x_{N-1}}(x_N)\\
&\ + \big( \chi_{N-1,n}(x_{N-2})  - \chi_{N-1,n}(x_{N-1})  + \xi_{N-2,n}( x_{N-2}) \cdot \Delta x_{N-2}\big)  d\pi_{\bar x_{N-2}}(x_{N-1})\\
&\ +\dots + \big( \chi_{2,n}(x_{1})  - \chi_{2,n}(x_{2})  + \xi_{1,n}( x_{1}) \cdot \Delta x_{1}\big)  d\pi_{\bar x_{1}}(x_{2}) d\pi^1(x_1).
\end{align*}
Since $\chi_{t+1,n}(x_{t})  - \chi_{t+1,n}(x_{t+1})  + \xi_{t,n}( x_{t}) \cdot \Delta x_{t} \le 0$ for all $t$, repeated application of Fatou's lemma allows $\limsup$ to continue to penetrate into the innermost integral. Using $\limsup (a_n + b_n) \le \limsup a_n + \limsup b_n$, this eventually yield
\begin{align*}
&\limsup_{n\to\infty} \int \sum_{t=1}^{N-1} \big{(}  \chi_{t+1,n}(x_t)  - \chi_{t+1,n}(x_{t+1})  + h_{t,n}(\bar x_t) \cdot \Delta x_t  \big{)} \, d\pi \\
&\le \int\dots \int \big( \chi_{N}(x_{N-1})  - \chi_{N}(x_{N})  + \xi_{N-1}( x_{N-1}) \cdot \Delta x_{N-1}\big)  d\pi_{\bar x_{N-1}}(x_N)\\
&\ + \big( \chi_{N-1}(x_{N-2})  - \chi_{N-1}(x_{N-1})  + \xi_{N-2}( x_{N-2}) \cdot \Delta x_{N-2}\big)  d\pi_{\bar x_{N-2}}(x_{N-1})\\
&\ +\dots + \big( \chi_{2}(x_{1})  - \chi_{2}(x_{2})  + \xi_{1}( x_{1}) \cdot \Delta x_{1}\big)  d\pi_{\bar x_{1}}(x_{2}) d\pi^1(x_1)
\end{align*}
for some $\xi_t(x_t)  \in \partial \chi_{t+1}(x_t)$ which is a limit point of the bounded sequence $\{ \xi_{t,n}(x_t)\}_n$. Substituting $\xi_t(x_t)$ back to $h_t(\bar x_t)$ in the above inequality yields
\begin{align*}
&\limsup_{n\to\infty} \int \sum_{t=1}^{N-1} \big{(}  \chi_{t+1,n}(x_t)  - \chi_{t+1,n}(x_{t+1})  + h_{t,n}(\bar x_t) \cdot \Delta x_t  \big{)} \, d\pi \\
&\le  \int \sum_{t=1}^{N-1} \big{(}  \chi_{t+1}(x_t)  - \chi_{t+1}(x_{t+1})  + h_{t}(\bar x_t) \cdot \Delta x_t  \big{)} \, d\pi.
\end{align*}
Combining the integrals in \eqref{a8} then yields the claim \eqref{claim1}, hence the theorem.
\end{proof}

\section*{Acknowledgments}
YS thanks colleagues at the Global Technology Applied Research center of JPMorganChase for their support and helpful discussions. \\
CC thanks colleagues at Quantitative Trading \& Research of JPMorganChase for many fruitful discussions.

\section*{Disclaimer}
This paper was prepared for informational purposes with contributions from the Global Technology Applied Research center of JPMorgan Chase \& Co. This paper is not a product of the Research Department of JPMorgan Chase \& Co. or its affiliates. Neither JPMorgan Chase \& Co. nor any of its affiliates makes any explicit or implied representation or warranty and none of them accept any liability in connection with this paper, including, without limitation, with respect to the completeness, accuracy, or reliability of the information contained herein and the potential legal, compliance, tax, or accounting effects thereof. This document is not intended as investment research or investment advice, or as a recommendation, offer, or solicitation for the purchase or sale of any security, financial instrument, financial product or service, or to be used in any way for evaluating the merits of participating in any transaction.

\end{document}